\documentclass[journal,10pt]{IEEEtran}

\usepackage{cite}
\usepackage{amsmath,amssymb,amsfonts}
\usepackage{algorithmic}
\usepackage{graphicx}
\usepackage{textcomp}
\usepackage{xcolor}
\usepackage{subfigure}
\usepackage{hhline}
\def\BibTeX{{\rm B\kern-.05em{\sc i\kern-.025em b}\kern-.08em
    T\kern-.1667em\lower.7ex\hbox{E}\kern-.125emX}}

\usepackage{tikz}
\newcommand*\circled[1]{\tikz[baseline=(char.base)]{
    \node[shape=circle,draw,inner sep=1pt] (char) {#1};}}

\usepackage{amsthm}
\newtheorem{definition}{Definition} [section]
\newtheorem{theorem}{Theorem}  [section]

\newtheorem{prob}{Problem} [section]  

\begin{document}

\title{Solving the Federated Edge Learning Participation Dilemma: A Truthful and Correlated Perspective
\thanks{This work is partially supported by the US NSF under grants CNS-2105004, OAC-1839746, DGE-2011117, CNS-185210, CNS-1560020. 
}}

\author{Qin Hu, Feng Li, Xukai Zou, Yinhao Xiao
\IEEEcompsocitemizethanks{
\IEEEcompsocthanksitem Qin Hu and Xukai Zou are with the Department of Computer and Information Science, Indiana University-Purdue University Indianapolis, IN, USA. \protect \\ Email: qinhu@iu.edu, xzou@iupui.edu
\IEEEcompsocthanksitem Feng Li is with the Department of Computer Information and Graphics Technology, Indiana University-Purdue University Indianapolis, IN, USA. \protect \\
Email: fengli@iupui.edu
\IEEEcompsocthanksitem Yinhao Xiao is with the School of Information Science, Guangdong University of Finance and Economics, Guangzhou, China. \protect \\
Email: 20191081@gdufe.edu.cn
}
}


\markboth{IEEE Transactions on Vehicular Technology,~Vol.~XX, No.~XX, XXX~2021}
{}

\maketitle

\begin{abstract}
An emerging computational paradigm, named \textit{federated edge learning} (FEL), enables intelligent computing at the network edge with the feature of preserving data privacy for edge devices. Given their constrained resources, it becomes a great challenge to achieve high execution performance for FEL. Most of the state-of-the-arts concentrate on enhancing FEL from the perspective of system operation procedures, taking few precautions during the composition step of the FEL system. Though a few recent studies recognize the importance of FEL formation and propose server-centric device selection schemes, 
the impact of data sizes is largely overlooked. 
In this paper, we take advantage of game theory to depict the decision dilemma among edge devices regarding \textit{whether to participate in FEL or not} given their heterogeneous sizes of local datasets. 
For realizing both the individual and global optimization, the server is employed to solve the participation dilemma, which requires accurate information collection for devices' local datasets. Hence, we utilize \textit{mechanism design} to enable truthful information solicitation. 
With the help of \textit{correlated equilibrium}, we derive a decision making strategy for devices from the global perspective, which can 
achieve the long-term stability and efficacy of FEL. For scalability consideration, we optimize the computational complexity of the basic solution to the polynomial level. 
Lastly, extensive experiments based on both real and synthetic data are conducted to evaluate our proposed mechanisms, with experimental results demonstrating the performance advantages.
\end{abstract}

\IEEEpeerreviewmaketitle

\section{Introduction}
As the amount of data generated at the network edge grows explosively, the conventional cloud computing can hardly afford the high bandwidth consumption or meet the low-latency requirement of smart applications on mobile devices, which leads to the emergence of mobile edge computing \cite{xiao2019edge}. Based on a recent report \cite{market}, edge computing has achieved a global market valuing \$3.6 billion in 2020, which is estimated to reach \$15.7 billion by 2025 with an annual growth rate of 34.1\%. Meanwhile, assisted by growing computation power of devices, various machine learning (ML) algorithms running on the edge becomes prevailing \cite{yang2020artificial}. Specifically, 
federated learning (FL) framework has been widely deployed in this scenario to address the privacy concerns of data owners, and thus being named as federated edge learning (FEL), which trains ML models by relying on the collaboration of distributed edge devices conducting local training and submitting model updates without explicitly disclosing their original data. 

Considering that battery-powered devices are only available to restrained resources for training ML models, the state-of-the-art studies mostly focus on optimized control during the learning process, including communication resource allocation and scheduling \cite{abad2020hierarchical,zeng2019energy,yang2019scheduling,amiri2020update,yang2020age,lim2021decentralized,lim2021dynamic},  FL algorithm upgrade \cite{wang2019adaptive,zhu2019broadband,yang2020federated,amiri2020machine,tran2019federated}, etc. 
Although the above optimization researches may perform well for the given set of participating devices, the composition step of FEL has long been overlooked. An inappropriate formation of the FEL system can result in low convergence speed and high computational cost. 
Being aware of this, a few recent studies aim at improving the system performance prior to the actual FEL procedures. Specifically, some researchers \cite{kang2019incentive,nishio2019client} select the appropriate set of devices to join FEL under communication or computational cost constraints, and others \cite{zhan2020incentive,zhan2020learning} design incentive mechanisms to elicit device participation in FEL based on the Stackelberg game. However, the former type of studies fail to consider the heterogeneity of devices' local dataset sizes and select devices mainly from the server's perspective, challenging the efficiency and sustainability of the formed FEL setup; 
while the latter ones usually assume the availability of perfect information and uniform data usefulness for all devices, 
	which may not hold in practice.

To address the above challenges, we study the problem of FEL system composition 
given their various sizes of individually collected datasets which are used for local learning. 
The rationale is that for devices with small local datasets, the necessary computational and communication consumption of joining FEL may not be compensated commensurately with the benefits brought by the finally returned ML models. This can, on the one hand, discourage the continuous contribution of devices in FEL, and on the other hand, degrade the cost-efficiency of the whole FEL system. Thus, the main issue for every device is to decide on \textit{whether to participate in an FEL task with the current local dataset?} 
From this perspective, the profit of devices is considered thoroughly, which can benefit the FEL systems with long-term stability and efficacy maintenance. 

However, there exists a major challenge for achieving optimal decision making that edge devices have no access to the complete information about the current decision strategies and results of other peering devices with respect to joining the current FEL or not. And in practice, it is the participation decisions of all devices that jointly affect the performance of the finally trained global model and further the overall benefit of devices. 
To solve this concern, we take advantage of the central location of the server in FEL to assist in calculating the optimal decision strategies for all devices that have various data sizes used for local training. Specifically, an optimization problem based on the \textit{correlated equilibrium} of the participation game is formulated and resolved, where the contribution of each device to the global model and the corresponding cost of resource consumption are examined for depicting individual profit. 

As a critical parameter involved in the derivation of the participation decision results, another challenge comes from the integrity of data size information submitted from devices. 
Considering that the local datasets are invisible to the edge server, devices may intentionally to report the false information of their local dataset sizes due to their intrinsic selfishness and potential attractions of making extra income with undisclosed data. This obviously brings huge difficulty to the collection of critical impact factors for decision making, which can further invalidate the design and derivation of aforementioned decision strategies. 
To get rid of this problem, we resort to \textit{mechanism design} for soliciting reliable submissions of local data sizes from devices without knowing other private information, such as the intentions of being malicious. By this means, an optimal game rule can be configured by the server and sent to each device, which is proved to be incentive-compatible, thus leading to the truthful reports from devices based on their real private information.


In summary, we make the following contributions (the first two are inheriting from the preliminary version \cite{globecom}):
\begin{itemize}
\item A \textit{participation game} is formulated to describe the intertwined conflict and collaboration among all devices, which models the influence of devices' local data sizes on the global model  and individual cost of joining FEL for payoff definition.
\item To prepare for solving the participation decision dilemma in FEL, we utilize \textit{mechanism design} theory to enable truthful data size information collection from all devices, where the detailed implementation process and optimality analysis are provided.
\item To jointly achieve the individual and global rationality, we design a \textit{correlated equilibrium} based scheme to address the participation dilemma, which is further improved by decreasing the running cost to a polynomial level.
\item Real-world dataset is employed to generate experimental parameters for practicality purpose, based on which extensive simulation experiments are conducted to evaluate both the participation decision scheme and the truthful data size information collection mechanism.
\end{itemize}

The rest of the paper is organized as follows. The related work are investigated in Section \ref{sec:related}. We present the problem formulation of participation decision in Section \ref{sec:formulation} and the specific solution in Section \ref{sec:solution}. 
The mechanism design for truthful data size collection is reported in Section \ref{sec:mechanism}. 
Experimental evaluation of our proposed schemes is in Section \ref{sec:experiment}, 
and the whole paper is concluded in Section \ref{sec:conclusion}.

\section{Related Work}\label{sec:related}

As a fresh concept emerging in the recent years, the current research on FEL performance improvement can be generally classified into in-operation optimization \cite{abad2020hierarchical,zeng2019energy,yang2019scheduling,amiri2020update,yang2020age,wang2019adaptive,zhu2019broadband,yang2020federated,amiri2020machine,tran2019federated,lim2021decentralized,lim2021dynamic} and beforehand planning \cite{kang2019incentive,nishio2019client,zhan2020learning,zhan2020incentive}.

For deploying FL at the network edge, substantial efforts haven been made on \textit{resource allocation} \cite{abad2020hierarchical,zeng2019energy,lim2021decentralized,lim2021dynamic}, \textit{transmission scheduling} \cite{yang2019scheduling,amiri2020update,yang2020age}, and \textit{learning algorithm refinement} \cite{wang2019adaptive,zhu2019broadband,yang2020federated,amiri2020machine,tran2019federated}. In \cite{abad2020hierarchical}, the problem of resource allocation was investigated in the proposed hierarchical FL framework with devices clustered to train models before reaching to the global aggregator. An energy-efficient radio resource allocation scheme devised in \cite{zeng2019energy} assigned more bandwidth to FEL participants with lower computing power for the sake of aggregation synchronization. 
Regarding a new paradigm named hierarchical federated learning using the  intermediate model aggregation to achieve higher communication efficiency, Lim \textit{et al.} \cite{lim2021decentralized,lim2021dynamic} studied the dynamic resource allocation with the help of game theoretical tools, including the evolutionary game, Stackelberg game, and auction.
Yang \textit{et al.} \cite{yang2019scheduling} studied three classes of transmission scheduling mechanisms, i.e., random, round-robin, and proportional fair, aiming at the optimal convergence rate for FL. To design better scheduling policies for FEL, two comparable studies  \cite{amiri2020update,yang2020age} were conducted to balance the channel conditions and local model updates of participated devices. Concentrating more on the FEL training performance, the authors \cite{wang2019adaptive} examined the convergence bound of gradient descent to inspire the best aggregation frequency of global model given limited resources. And three similar research \cite{zhu2019broadband,yang2020federated,amiri2020machine} took advantage of the \textit{over-the-air computation}, which relies on the feature of multiple access channel superposition, to realize more efficient global aggregation. Further, trade-offs between FL training accuracy, latency and devices' energy cost were  accomplished in \cite{tran2019federated} .

Rather than optimizing critical steps during the learning process, several studies indicate that taking precautions in composing the FEL system can bring more benefits, where \textit{device selection} for filtering out unqualified participants and \textit{incentive mechanism design} for attracting participation of devices are two main types of research. For device selection in FEL, a reputation based scheme to identify trustworthy devices was proposed in \cite{kang2019incentive}, after which an incentive mechanism based on the contract theory was designed to encourage the submissions of high-quality data; and a novel scheme named FedCS \cite{nishio2019client} was devised to embrace an extra step of device selection considering the heterogeneity of computation and communication resources. While for incentive mechanism design, Zhan \textit{et al.} \cite{zhan2020learning,zhan2020incentive} utilized game theory and deep reinforcement learning (DRL) to derive the optimal contribution for devices, i.e., the amount of contributed data in \cite{zhan2020learning} and the devoted CPU-cycle frequency in \cite{zhan2020incentive}, and the best payment policy for the edge server. 

It is clear that most of the existing work concentrate on the optimization during FEL process without considering the importance of FEL organization. Although a few recent studies design server-centric device selection schemes and incentive mechanisms before the beginning of FEL, they either fail to consider the impact of data sizes from devices on learning performance, or have strong assumption on complete information availability and uniform data usefulness. To overcome these shortages, we take advantage of mechanism design theory to collect truthful data size information from devices, based on which a correlated participation decision scheme is proposed for devices with various generated data sizes.

\section{Problem Formulation}\label{sec:formulation}

\subsection{System Model}

\begin{figure}[t]
\centering
\includegraphics[width=0.35\textwidth]{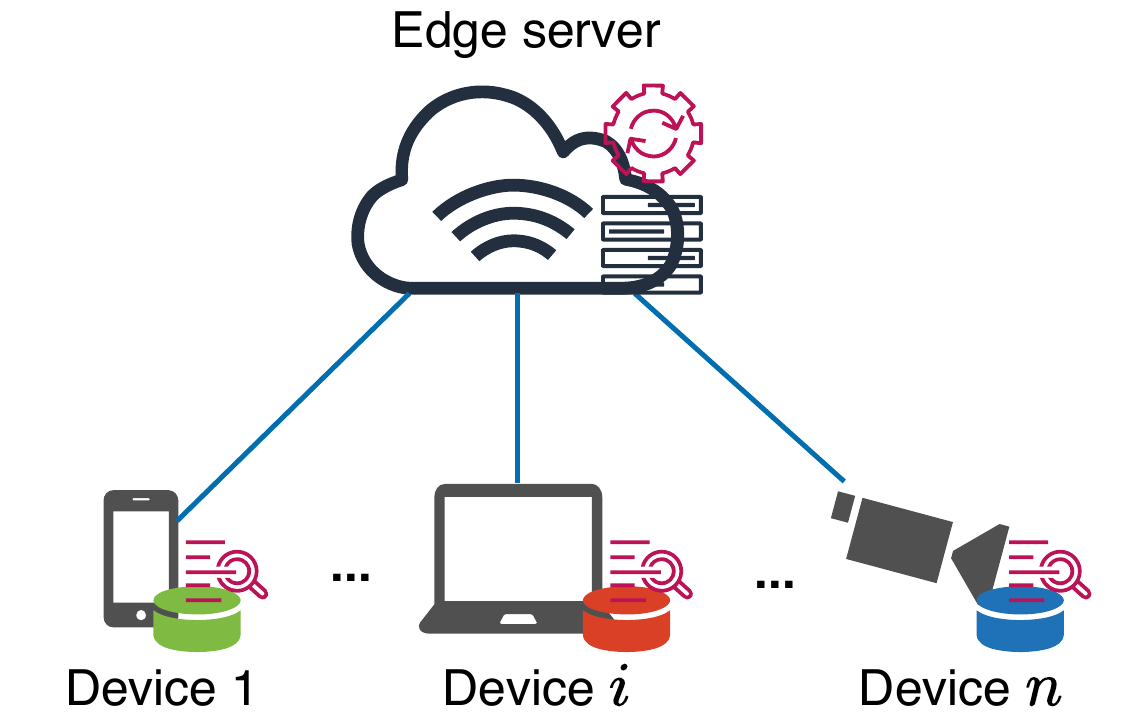}
\caption{The FEL system.}
\label{fig:system}
\end{figure}

As illustrated in Fig. \ref{fig:system}, an FEL system usually consists of one edge server and $n$ edge devices, denoted as $\mathbf{d} = \{d_1,d_2,\cdots,d_n\}$, with their available computation and communication resources registered on the server.
All participants in this FEL system work collaboratively to train an accurate ML model so as to to provide real-time and smart responses to devices. Without the loss of generality, the ML model trained in FEL is assumed to be a convolutional neural network (CNN) based classifier. Here any device $d_i$ desires to receive a robust ML model from the edge server after the whole FL training task finishes, where the model is trained based on the heterogeneous data from all devices; 
while the edge server coordinates with all devices using wireless networks to send, receive and aggregate model parameters. 

The above collaboration paradigm seems to operate well most of the time; however, in the case of some devices generating too less amount of data to conduct local model training, the FEL system formed above may not work efficiently. For example, surveillance cameras can only capture still pictures during late night, which cannot serve as a local dataset for model training. In fact, the sizes of devices' local datasets in FEL can be diverse, where some of them may fail to satisfy the minimum requirement for local model training, such as only a few data samples. It is worth noting that even for the case of non independent and identically distributed (non-IID) data in FEL, a local dataset with too-small size is still not favorable for conducting model training at the corresponding device. 
Without taking any actions, directly involving these devices in FEL can lead to huge wastes of computation and communication resources for both the server and devices while bring no benefit to the ML model. In the worst case, this can even degrade the FEL efficiency since some meaningless local model updates would also be aggregated to update the global model. 
For simplicity, we can summarize this problem as 
\begin{prob}\label{prob:1} 
How should each device $d_i$ decide on whether to participate in an FEL given its local data size $s_i$?
\end{prob}

\subsection{Participation Game Formulation}
We denote the participation decision of $d_i$ as $p_i \in \{0,1\}$, with 1 (or 0) denoting that $d_i$ decides (or not) to join this FEL system. Then the decision vector indicating all devices' decisions can be expressed as $\mathbf{p}=(p_1,p_2,\cdots,p_n)$.

\subsubsection{Total Incentive}
Due to the fact that the final ML model is holistically trained with the data from participated devices, the model performance is jointly affected by the participation decisions of all devices. 
To guarantee the long-term liveliness of the FEL system, we consider that the server will provide a total incentive (e.g., monetary reward), denoted as $\pi$, based on the quality of the finally returned classification model. To be specific, we can define $\pi$ as
\begin{equation}\label{eq:total}
\pi(\mathbf{p}) = 
\begin{cases}
0,~ &\sum_{i=1}^n p_i = 0, \\
\alpha (1 - a(\sum_{i=1}^n p_is_i)^{-b}), ~ &o/w,
\end{cases}
\end{equation}
where $\alpha > 0$ is a system parameter and $a,b\geq 0$ are tuning scalars of the power law function modeling the classification error. In particular, $\sum_{i=1}^n p_is_i$ in \eqref{eq:total} reflects the effectively total size of the training data contributed by all devices. And the power law modeled error is inspired by \cite{chen2018my, johnson2018predicting}, depicting the non-linear relationship between the classification error and training data size, which reflects the feature of an increasing total size of training data $\sum_{i=1}^n p_is_i$ corresponding to a lower classification error. The lower error can offer the FEL system a better ML model, bringing a higher incentive for all devices. 

\subsubsection{Participation Income and Cost}\label{subsubsec:income}
Then, the incentive of every device can be fairly determined according to their respective contribution. 
Although there exist some studies proposing fancy mechanisms to quantify the contribution of FL clients \cite{wang2019measure,song2019profit}, we consider the reward any device $d_i$ will receive is proportional to the size of its local dataset, which has also been employed in some recent work \cite{zhan2020learning}. Specifically, we define the reward of $d_i$ as follows:
 
%
%
\begin{equation}\label{eq:reward}
\Phi_i (\mathbf{p}) =  \frac{p_i s_i}{\delta + \sum_{j=1}^n p_j s_j} \pi(\mathbf{p}),
\end{equation}
where $\delta$ is a positive but small number close to zero, which is used to handle the special case of no device participating FEL. In detail, if there is no device contributing to the training process, we have $\sum_{j=1}^n p_j s_j = 0$, and the existence of $\delta$ makes the the denominator of \eqref{eq:reward} never be zero and the definition of reward meaningful. While once any device joins a round of FEL, the impact of $\delta$ on reward distribution will be trivial since the data size is much larger.


For any device $d_i$ contributing to FEL, certain amounts of computation and communication resources will be consumed. According to \cite{tran2019federated}, the computing cost will be positively proportional to the local training data size $s_i$, while the communication cost is determined by the model size and respective wireless channel conditions, denoted as $w_i$.  
Thus, we can calculate the cost for a participating device $d_i$ as
\begin{equation}
\Omega_i = \beta_i s_i + \gamma_i w_i,\nonumber
\end{equation}
where $\beta_i,~\gamma_i >0$ are constant scalars.

\subsubsection{Profit and Game Definitions}
Based on the previous two subsections, the profit of any device $d_i$ can be defined as follows.

\begin{definition}[Device's Profit]\label{def:profit}
Given the decision vector $\mathbf{p}$, the profit of device $d_i$ is 
\begin{equation}
V_i(\mathbf{p}) =\Phi_i - p_i \Omega_i.\nonumber
\end{equation}
\end{definition}
Based on the above definition, the profit $V_i$ consists of the income gained from the server and the participation cost consumed for local model training and updating.

As mentioned in Section \ref{subsubsec:income}, $\Phi_i$ is jointly affected by the participation decisions of all devices. Meanwhile, the impact of participation cost $\Omega_i$ on the profit $V_i$ is individually influenced by the decision $p_i$. Therefore, the profit of each device $d_i$ is not only determined by its own participation decision $p_i$ but also collectively decided by the decisions of other devices, which is denoted as $\mathbf{p}_{-i}=(p_1,\cdots,p_{i-1},p_{i+1},\cdots,p_n)$ for simplicity. And the following \textit{Participation Game} can thus model the intertwined relationship among all devices.

\begin{definition}[Participation Game]\label{def:game}
In this participation game, any device $d_i$ as a game player chooses a strategy $p_i$ regarding whether to participate in an FEL system to get a payoff of $V_i(\mathbf{p})$.
\end{definition}

In the participation game, any rational player $d_i$ desires to obtain the maximum profit $V_i(\mathbf{p})$. 
However, as we can see in the definition of $V_i(\mathbf{p})$, without knowing others' decisions $\mathbf{p}_{-i}$, no player can easily achieve this goal in an individual manner via choosing an optimal $p_i$. To get rid of this dilemma, the edge server, residing in a core position coordinating with all devices in FEL, presents great potential to address this challenge from a global viewpoint, which requires the reliable collection of necessary information from devices. Apart from the computation and communication resource parameters that are critical for FEL training and updating processes, the local data size $s_i$ turns into an important factor for participation decision making as presented in Problem \ref{prob:1}. 

However, since the local datasets are not visible to the edge server, 
there might exist some malicious devices deceiving the edge server and other peering devices via providing fake information of $s_i$. In particular, edge devices may intentionally report either a lower or higher value than the real one of $s_i$. Here the fake lower value could enable selfish devices to avoid the participation of FEL for resource saving while the intentionally-fabricated higher $s_i$ may empower unqualified devices to obtain abundant intermediate learning results for other uses. Both types of fake information submission can benefit the malicious devices at the cost of damaging the interests of the edge server and other benign devices. 
Therefore, we summarize this challenge as another problem:
\begin{prob}\label{prob:2}
How can the edge server elicit the truthful information of local data size $s_i$ from any device $d_i$?
\end{prob}

In the following sections, we will answer the above two problems in reverse order since the accurate information collection acts as a foundation for the optimal participation decision making. In detail, we first solve Problem \ref{prob:2} in Section \ref{sec:mechanism}, based on which Problem \ref{prob:1} can be addressed in Section \ref{sec:solution}.

\section{Mechanism Design for Truthful Data Size Information Solicitation}
\label{sec:mechanism}

As mentioned above, edge devices may submit incorrect information about their local data sizes, which can severely affect the next-step calculation of participation decisions for all devices. 
To characterize their malice in this process, we define a probability $\Theta_i\in [0,1]$ for each device $d_i$. 
In this section, we resort to \textit{mechanism design theory} \cite{hurwicz2006designing} to eliminate this undesirable phenomenon via enforcing their truth-telling behaviors during the step of submitting $s_i$.

Technically, the mechanism design theory aims to find solutions  of incentive schemes to achieve desired goals in private-information games using an objective-first manner, which fundamentally relies on the sweeping result of \textit{revelation principle}. This principle advocates that for any incomplete-information game, i.e., Bayesian game, each Nash equilibrium is corresponding to another direct equilibrium achieved by an incentive-compatible mechanism where every player honestly reports the private information. Thus, we can easily solve the Bayesian-Nash equilibrium of the mechanism design game with incomplete information of the opponent's real strategy, via assuming every player tells the truth if we can guarantee the incentive-compatibility of the proposed mechanism.

In our scenario, devices are reluctant to reveal the real $s_i$ with an probability $\Theta_i$ as the fabricated data size information can bring them extra profit. 
Since $\Theta_i$ heavily impacts the individual interest of every device, it is obvious that no one would like to share this private information to others, including the edge server. Hence, for the mechanism design game of data size information collection, the objective of the server is to elicit $s_i$ from devices based on their truthful malice using the power of reward policy development, without asking for their private information $\Theta_i$.

In the following, we will elaborate the mechanism design problem and results to help the edge server collect data size information from devices, 
thus facilitating the subsequent participation decision making process discussed in Section \ref{sec:solution}. 
Considering that the designed mechanism will be conducted between the server and every device, we omit the subscript $i$ for brevity.

\subsection{Utility Functions}\label{subsec:utility}
As mentioned above, the device needs to submit $s$ for obtaining decision making information and the finally well-trained ML model to better serve users. Since these subsequent outcomes are sent back from the edge server and significantly affect the benefit of the device, we can regard this procedure as a reward policy, denoted by a coefficient $r$,  
which is determined by the edge server. And for the device, the controllable strategy is its submitted data size information $s$. 

Based on these definitions, we can express the expected utility of the device, denoted by $U_d$, during a time period $[0,T]$ as 
\begin{equation}\label{eq:ud}
U_d (r,s) = \int_0^T (rs + x(\Theta,s)) dt,
\end{equation}
where the first term of the integrand indicates the normal reward of the device obtained from the action of submitting the data size information to the server, while $x(\Theta,s)$ denotes the extra profit that the device can harvest via maliciously reporting untruthful data size information. As the extra profit is positively related to both the data size $s$ and the probability of exerting malice $\Theta$, here we can define $x(\Theta,s)=A_d \Theta s + B_d$ as an example, with $A_d,B_d$ being non-negative scalars. 

Similarly, the  expected utility of the edge server, denoted by $U_e$, during $[0,T]$ can be represented by
\begin{equation}\label{eq:ue}
U_e (r,s) = \int_0^T (R(r,s) - y(\Theta,s)) dt,
\end{equation}
where $R(r,s)$ is the normal reward that the server can obtain by collecting the data size information from the device, and $y(\Theta,s)$ is the potential loss of the server when the device maliciously hides the real data size information to obtain extra profit, 
defined as $y(\Theta,s) = A_e \Theta s r + B_e$ with $A_e,B_e \geq 0$ being scalar parameters. Note that  the definition of $y(\Theta,s)$ is slightly different from that of $x(\Theta,s)$ since the impact of reward coefficient on the server's loss is considered here. 

For the server's reward $R(r,s)$, considering that the larger the data size of the device, the higher the benefit of the server, while an unfitted reward coefficient can decrease the benefit of the server, we model it as:
\begin{equation}\label{eq:R}
R(r,s) = \frac{\sigma}{1+e^{-(s-s_0)}}-\rho(r-r_0)^2,
\end{equation}
where $\sigma,\rho>0$ are constant parameters; $s_0$ and $r_0$ are respectively the expected values of the data size and reward coefficient according to the historical and global information. In particular, a sigmoid function is employed in the first part of \eqref{eq:R} to describe the influence of $s$ on the server's reward. 
To be specific, when the submitted data size $s<s_0$, the reward of the server is limited but its gradient gradually increases; while if $s>s_0$, the reward asymptotically approaches the largest value with a decreasing slope. This corresponds to the fact that a smaller dataset is definitely not preferred for the server, but a too-large dataset suffers from diminishing marginal contribution to the server's reward. Besides, the latter part of \eqref{eq:R} captures the feature that the server will not assign a too-high or too-low $r$ to the device, because a higher one might decrease the server's reward if the device contributes less in FEL while a lower value of reward coefficient can hurt the device's interest and discourage its future contribution.

\noindent \textbf{Remark:} For \eqref{eq:ud} and \eqref{eq:ue}, it is worth noting that the device's normal reward $rs$ is not explicitly deducted from the server's reward, where the underlying reason is that the reward of the device is not directly distributed by the server but closely related to other environment parameters, such as the owner satisfaction of the device.

\subsection{Mechanism Design Process}\label{subsec:process}
With the above-defined utilities and the strategies of both the device and server, we can depict the interaction process in the mechanism design game in Fig. \ref{fig:mechanism} and summarize the general steps as follows:
\begin{itemize}
\item The server sends a game rule (i.e., mechanism) $r^*(s)$ to the device, which is usually designed to maximize its expected utility $U_e$.
\item With the received $r^*(s)$, the device can derive the best action $s^*$ based on the true private information $\Theta$. Typically, the devices determines $s^*(\Theta)$ with the goal of maximizing the expected utility $U_d$.
\item According to the derived $s^*$ and the corresponding $r^*(s^*)$, the device can make a decision on whether to accept this game rule. If it turns out to be beneficial, the device will send back $s^*$ to the server; if not, the device will keep silent.
\item Once the server receives $s^*$ in a given time limit, the specific value of $r^*$ can be obtained. And then the server will proceed to the next step of calculating the decision making vector for all devices as mentioned in Section \ref{sec:solution}.
\end{itemize}

\begin{figure}[h]
\centering
\includegraphics[width=0.45\textwidth]{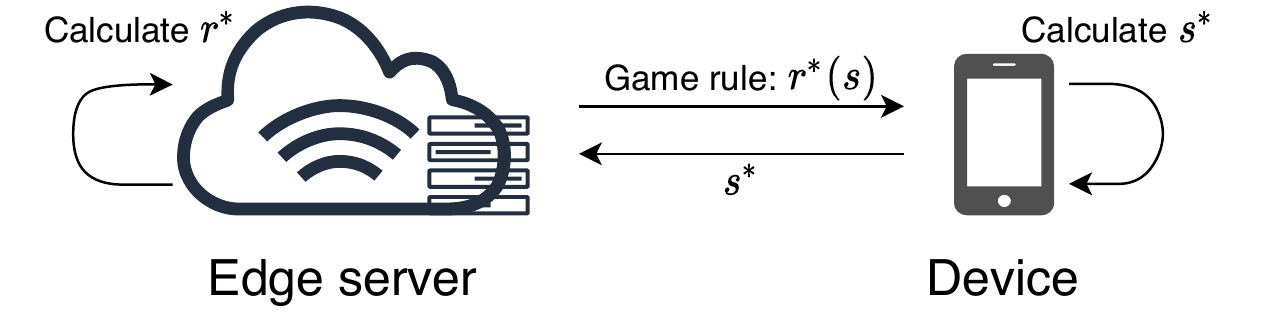}
\caption{Mechanism design process.}
\label{fig:mechanism}
\end{figure}

It is worth noting that the device may strategically report another $\tilde{s}$ based on the fake private information $\tilde{\Theta}$ in the third step. However, under the function of revelation principle, the incentive compatibility of the proposed mechanism $r^*(s)$ will enable the device to find out that truthful report is exactly the optimal choice, which will be rigorously proved and analyzed in Section \ref{subsec:analysis}.

\subsection{Derivation of Optimal Strategies}\label{subsec:strategy}
To further study the results of the proposed mechanism design process, we reveal the optimal strategies of both sides. 
Based on the first step in the aforementioned mechanism design process, the edge server needs to calculate a game rule $r^*(s)$ maximizing the utility $U_e$. Since the game rule is a function instead of a pure variable, we may utilize the calculus of variations method to derive $r^*(s)$. In detail, 
we first denote the integrand part in parentheses of \eqref{eq:ue} as $H_e=R(r,s)-y(\Theta,s)$, then $r^*(s)$ can be derived by solving the associated Euler-Lagrange equation
\begin{equation}
\frac{\partial H_e}{\partial r} - \frac{d}{dt} \frac{\partial H_e}{\partial r'} = 0,
\end{equation}
under the condition $\frac{\partial^2 H_e}{\partial r^2}<0.$ 
Since $H_e$ is not explicitly related to $r'$, the above equation turns into $\frac{\partial H_e}{\partial r} = 0$, which derives 
\begin{equation}\label{eq:rule}
r^*(s) = r_0 - \frac{A_e \Theta s}{2\rho}.
\end{equation}
 Meanwhile, we can calculate $\frac{\partial^2 H_e}{\partial r^2}=-2\rho$ which is obviously negative since $\rho>0$. Therefore, we can confirm that the above $r^*(s)$ can maximize $U_e$.

Using the similar method, we can derive $s^*$ with the calculated $r^*(s)$ as $$s^* = \frac{\rho(r_0+A_d\Theta)}{A_e \Theta},$$ which maximizes $U_d$  defined in \eqref{eq:ud} under the condition $\frac{\partial^2 H_d}{\partial s^2}<0$. In fact, denoting the integrand of $U_d$ as $H_d$, we can calculate $ \frac{\partial^2 H_d}{\partial s^2} = -\frac{A_e \Theta }{\rho} <0$ because $A_e,\Theta,$ and $\rho$ are all positive. 

\subsection{Truthfulness Analysis}\label{subsec:analysis}
To investigate the effectiveness of the mechanism design for data size information collection, 
we theoretically analyze 
the truthfulness of the device in the step of submitting $s$ under the function of the server's designed game rule $r^*(s)$. In particular, we investigate that whether the proposed $r^*(s)$ satisfies the \textit{incentive-compatibility principle}. To be specific, as defined in \cite{algorithmic2007}, a mechanism is incentive-compatible if only exerting actions based on the real preferences can realize the optimal outcome for every player. In our case, this principle compels that only when the device tells the truth about $\Theta$, can its expected utility be maximized, i.e., meeting the incentive-compatibility constraint $U_d(r(s),s(\Theta))\geq U_d(r(\hat{s}),\hat{s}(\hat{\Theta}))$ where $\hat{\Theta}$ represents any possible value of $\Theta$ and $\hat{s}$ is the corresponding report. By this means, the data size information $s$ can be solicited according to the truthful intention of the device behaving maliciously. 

The detailed analysis about the incentive compatibility of the designed game rule $r^*(s)$ is presented in the following theorem.

\begin{theorem}
The server's proposed game rule $r^*(s)$ shown in \eqref{eq:rule} is incentive-compatible.
\end{theorem}
\begin{proof}
Assuming that the device pretends to be malicious with a different probability $\tilde{\Theta}$ during the step of submitting its local data size, which is not equal to $\Theta$.  Thereby, the device will send back a newly best strategy $\tilde{s}^*$ based on $\tilde{\Theta}$ once receiving and accepting the game rule $r^*(s)$, where $\tilde{s}^* \neq s^*$. However, as we mentioned earlier, the device derives the best strategy $s^*$ for maximizing its expected utility, so there will exist $ U_d(s^*) > U_d(\tilde{s}^*)$. That is to say, if the device delivers $\tilde{s}^*$ based on $\tilde{\Theta}$, there will be some room left to further increase its expected utility, which contradicts the device's goal of submitting the optimal strategy for maximizing $U_d$. Therefore, a rational device will only calculate the optimal strategy based on the real private information, demonstrating the incentive compatibility of the devised game rule.
\end{proof}

\section{Correlated Equilibrium for Participation Decision}\label{sec:solution}

As mentioned in Section \ref{sec:formulation}, since it is challenging for devices to achieve individual optimum in the participation game, 
we consider to address Problem \ref{prob:1} from a global perspective. Specifically, the edge server located in FEL center helps the calculation after collecting truthful size information about local datasets from devices as reported in the above section. In this section, the main concept of correlated equilibrium in the participation game is first illustrated to realize each player's individual rationality, based on which a Global Profit Maximization (GPM) problem is defined to realize global optimum in Section \ref{subsec:game}.  
Then we conduct rigorous analysis on the computational cost of the basic solution, which is reduced by designing an approximation algorithm in Section \ref{subsec:improve}.

\subsection{Basic Solution}\label{subsec:game}


Formally, we define the strategy space of any player (device) in the participation game as $\mathcal{P}=\{0,1\}$ and its size is $P=2$. Thus, the calculated optimal decision vector $\mathbf{p}$ will come from the space $\mathcal{P}^n$, i.e., all the possible combinations of devices' participation decisions. Recall that any player $d_i$ in the participation game has an objective of obtaining 
the optimal profit $V_i(\mathbf{p})$ defined in Definition \ref{def:profit}. To indicate different impacts of strategies on each device's profit,  we rewrite $V_i(\mathbf{p})$ into $V_i(p_i,\mathbf{p}_{-i})$. According to \cite{algorithmic2007}, $V_i(p_i,\mathbf{p}_{-i})$ can be optimized at the \textit{correlated equilibrium} of the participation game, which can be defined as follows.

\begin{definition}[Correlated Equilibrium of the Participation Game]\label{def:correlated}
In the participation game,  a probability distribution over the space $\mathcal{P}^n$, denoted as $G(\mathbf{p})$, is termed as a correlated equilibrium iff $G(\mathbf{p})$ makes the following condition hold for any strategy $p_i,p'_i \in \mathcal{P}$,
\begin{equation}
\sum_{\mathbf{p}_{-i}\in \mathcal{P}^{n-1}} G(p_i,\mathbf{p}_{-i}) \Big( V_i(p_i,\mathbf{p}_{-i}) - V_i(p'_i,\mathbf{p}_{-i}) \Big) \geq 0. \nonumber
\end{equation}
\end{definition}

In light of the above definition, we can tell that under the correlated equilibrium $G(\mathbf{p})$, there is no player with the motivation to deviate from the assigned strategy $p_i$ given other devices' strategies $\mathbf{p}_{-i}$. In other words, only by playing the game with the strategy $p_i$ in $\mathbf{p}$ sampled from the correlated equilibrium $G(\mathbf{p})$, can any player $d_i$ optimizes the profit. 

In fact, multiple different correlated equilibria might meet the above condition. 
Considering that the edge server in FEL derives the participation decisions based on the global utility, i.e., $\sum_{\mathbf{p}\in \mathcal{P}^{n}} G(\mathbf{p}) \sum_{i=1}^n V_i(\mathbf{p})$, we can calculate the best correlated equilibrium of the participation game via solving the following \textit{Global Profit Maximization (GPM) problem}.

\textbf{GPM Problem:}
\begin{align}
\max:~& \sum_{\mathbf{p}\in \mathcal{P}^{n}} G(\mathbf{p}) \sum_{i=1}^n V_i(\mathbf{p}) \label{eq:objective} \\ 
\mathrm{s.t.}:~& G(\mathbf{p})\geq 0, ~\forall \mathbf{p} \in \mathcal{P}^n, \label{eq:bigzero}\\
& \sum_{\mathbf{p}\in \mathcal{P}^n} G(\mathbf{p})=1, \label{eq:equalone}\\
\sum_{\mathbf{p}_{-i}\in \mathcal{P}^{n-1}} &G(p_i,\mathbf{p}_{-i}) \Big( V_i(p_i,\mathbf{p}_{-i}) - V_i(p'_i,\mathbf{p}_{-i}) \Big) \geq 0, \notag \\ & \forall p_i,p'_i \in \mathcal{P}. \label{eq:correlated}
\end{align}

Clearly, the optimization variable is $G(\mathbf{p})$ in the above GPM problem, with the optimization object  in \eqref{eq:objective} to maximize the overall expected profit for all devices in the participation game.  The first constraint \eqref{eq:bigzero} is a natural requirement for the probability distribution, the second one \eqref{eq:equalone} represents that the sum of all probability distribution over the strategy space $\mathcal{P}^n$ equals to 1, and the last one \eqref{eq:correlated} is the definition of correlated equilibrium for individual profit maximization. 


 \subsection{Improvement of Computational Cost}\label{subsec:improve}

It can be seen that the above GPM problem is actually a linear programming problem in terms of $G(\mathbf{p})$ and might be addressed with several efficient methods, e.g., simplex and interior-point algorithms. Nevertheless, the overall computational cost of the existing solutions is proportionally related to the number of constraints and variables, making the direct method of applying existing algorithms to solve the GPM problem inefficient since  the number of constraints is $P^n + P^2n +1$ and the number of variables ($G(\mathbf{p})$) is $P^n$. 
For clarity, we analyze the computational cost of direct applying existing linear programming algorithms on the GPM problem in the below theorem.

\begin{theorem}\label{thrm:direct}
The computational complexity increases exponentially with the number of devices $n$ in the case of directly using existing linear programming solutions to solve the GPM problem.
\end{theorem}
\begin{proof}
Considering that the number of optimization variables is $P^n$ and that of constraints is $P^n + P^2n +1$, we can derive the valued results in the GPM problem as $2^n$ and $2^n+4n+1$ due to the strategy space has a size of $P=2$ in our scenario. 
Thus, even adopting efficient solutions that can solve the linear programming problem in polynomial time, when we directly apply them on the GPM problem, the computational complexity turns out to be $\mathcal{O}(2^n)$. 
\end{proof}

In practical, the number of devices $n$ can be large and extend substantially sometimes, making the straight application of existing polynomial-time algorithms without any change on our problem inefficient. To solve this challenge, we enhance the above basic solution to decrease the computational cost to an acceptable level in the following. 

Specifically, to decrease the computational complexity of solving the GPM problem, the key step is to prevent the changing trend that the numbers of variables and constraints increase exponentially with the number of devices. In this case, a rough idea to enhance the basic solution is controlling the quantities of variables and constraints varying with $n$ in a polynomial manner, thus approaching an overall polynomial complexity. Generally, we divide the whole set of edge devices into several subsets, and then there exists a smaller-scale participation decision problem, i.e., sub-GPM (SGPM) problem, for devices in each subset. 
Suppose that $\xi$ small device subsets $\{\mathbf{d}_1,\cdots,\mathbf{d}_{\xi}\}$ are formed in light of the communication order of devices reporting $s_i$ to the server, with the size of each subset being upper-bounded by $\bar{n} \ll n$. For simplicity, we denote the quantity of devices in subset $\mathbf{d}_j$ as $n_j,~j \in \{1,2,\cdots,\xi\}$. 
And now, we can have the SGPM problem as follows,

\textbf{SGPM Problem: }
\begin{align}\label{eq:sgpm}
\max:~& \sum_{\mathbf{p}_s\in \mathcal{P}^{n_s}} G(\mathbf{p}_s) \sum_{i=1}^{n_s} V_{s,i}(\mathbf{p}_s)  \notag \\ 
\mathrm{s.t.}:~& G(\mathbf{p}_s)\geq 0, ~\forall \mathbf{p}_s \in \mathcal{P}^{n_s}, \notag \\
& \sum_{\mathbf{p}_s\in \mathcal{P}^n_s} G(\mathbf{p}_s)=1, \notag \\
\sum_{\mathbf{p}_{s,-i}\in \mathcal{P}^{n_s-1}} &G(p_{s,i},\mathbf{p}_{s,-i}) \Big( V_{s,i}(p_{s,i},\mathbf{p}_{s,-i}) \notag \\ & - V_{s,i}(p'_{s,i},\mathbf{p}_{s,-i}) \Big) \geq 0,  \forall p_{s,i},p'_{s,i} \in \mathcal{P}^{n_s}. \notag
\end{align}
According to the solutions of all SGPM problems in device subsets, i.e., $G_j(\cdot)$'s, an approximate answer of the GPM problem can be derived. 

Similarly, we can analyze the computational complexity of the enhanced solution.
\begin{theorem}
If the number of subsets $\xi$ polynomially increases with the number of devices $n$, the GPM problem can be resolved with the enhanced solution in $\mathcal{O}(n)$ time complexity. 
\end{theorem}
\begin{proof}
After splitting, the size of decision vector in any SGPM problem is $2^{n_j}$, and thus the number of variables and that of constraints become $\mathcal{O}(2^{n_j})$, leading to the computational cost of solving the SGPM using existing efficient linear programming algorithms $\mathcal{O}(2^{n_j})$ as well. While considering that $n_j$ has an upper bound of $\bar{n} \ll n$, we can have the complexity of solving each SGPM problem as $\mathcal{O}(2^{\bar{n}})$ which is unchanged with $n$. Thus, if $\xi$ changes with $n$ polynomially, we can derive that the overall complexity of solving the GPM problem is $\xi\cdot \mathcal{O}(2^{\bar{n}})$, which is clearly polynomial in $n$.
\end{proof}

\subsection{Overview of the Decision Making Process}\label{subsec:process}
As shown in Fig. \ref{fig:process}, we can illustrate the specific working process for participation decision making in FEL using our proposed solutions. To begin with, the edge server registers necessary equipment information during the registration step, e.g., computation performance and communication condition parameters. When an FEL task arrives, the device submits the size information of its locally collected data, i.e., $s_i$, following the well-designed game rule as proposed in Section \ref{sec:mechanism} (Step \circled{1}); after the server receives all the size information from devices to solve the GPM problem, the optimal decision probability distribution $G^*(\mathbf{p})$ can be calculated (Step \circled{2}) and sent back to the device (Step \circled{3}); based on the received $G^*(\mathbf{p})$, $d_i$ can obtain the optimal $p^*_i$ (Step \circled{4}) which might also be updated to the server (Step \circled{5}) for better arranging the next FEL procedures. If a round of FEL task finishes, the edge server will share the final ML model to all connected devices for providing better services to users.

\begin{figure}[h]
\centering
\includegraphics[width=0.45\textwidth]{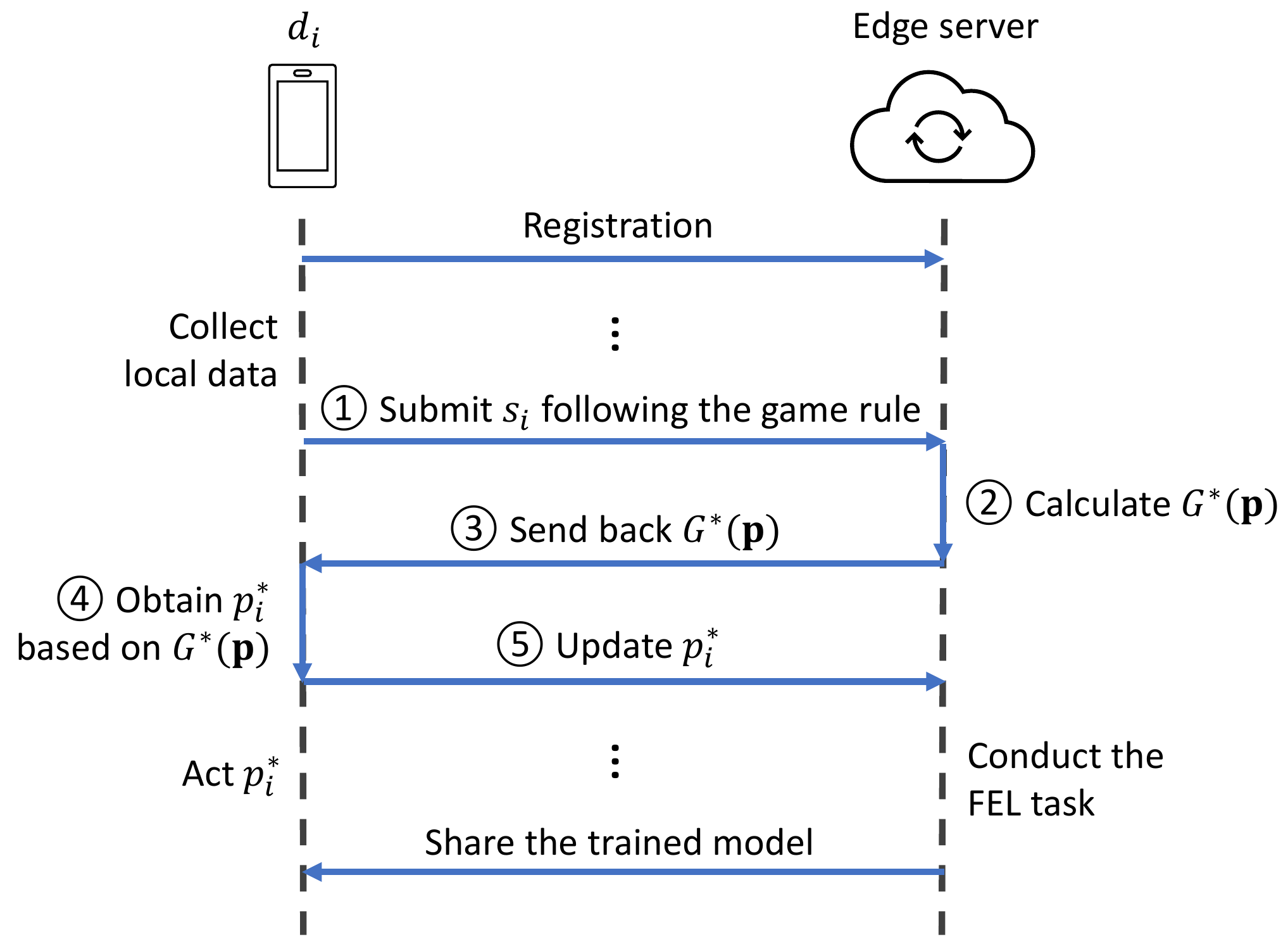}
\caption{Interaction process of our proposed solution.}
\label{fig:process}
\end{figure}

\section{Experimental Evaluation}\label{sec:experiment}
In this section, we conduct experiments to validate the effectiveness and efficiency of our proposed participation decision making scheme with the function of mechanism design for eliciting truthful local data sizes from devices. 
All experiments are implemented on a desktop with 3.59 GHz AMD six-core processor and 16 GB memory running Windows 10 OS. For ML related experiments, Python 3.6 is utilized for implementation and Matlab R2020a is used for others. Note that all the experimental results reported in this section are averaged from 30 times of repeated experiments for statistical confidence.


\subsection{Truthful Data Size Information Solicitation}\label{subsec:exp_truthful}

We simulate the mechanism design process for truthful data size information collection as presented in Section \ref{sec:mechanism}. To conduct the simulation experiments, we set the default values of parameters related to the device as $\Theta = 0.5$ and $A_d=B_d=1$, while those related to  the server as $A_e=B_e=1$, $r_0=50$, $s_0=500$, $\sigma=10^{5}$, and $\rho=10$, unless they are specified otherwise.

To investigate the impacts of the device's private information $\Theta$ on the maximized utility of the device itself and that of the server, we change $\Theta$ from 0.1 to 1 and the results are reported in Fig. \ref{fig:utility_theta}. It is obvious that with an increasing $\Theta$, both the device and the server obtain decreasing values of their maximized utilities. This is definitely reasonable and achieves our expectation. For a device with a higher probability of being malicious to submit the data size information, the power of the designed game rule enforces a lower expected utility for it; at the same time, the profit of the server will also reduce. It is worth noting that the difference between these two curves is because the derived $r^*(s)$ is linear to $\Theta$ while $s^*$ is not, making the optimized utilities follow the similar trends. 
Then we change the device's scalar parameters $A_d$ and $B_d$ from 0 and 1 as well to study their impacts on the maximized utilities of both sides. From Fig. \ref{fig:utility_Ad_Bd}, one can figure out that the impact of $A_d$ is obviously larger than that of $B_d$ even both are in the same range, where the maximized $U_d$ is increasing with $A_d$ while the maximized $U_e$ decreases with $A_d$. 

\begin{figure}
\subfigure{
\includegraphics[width=0.23\textwidth]{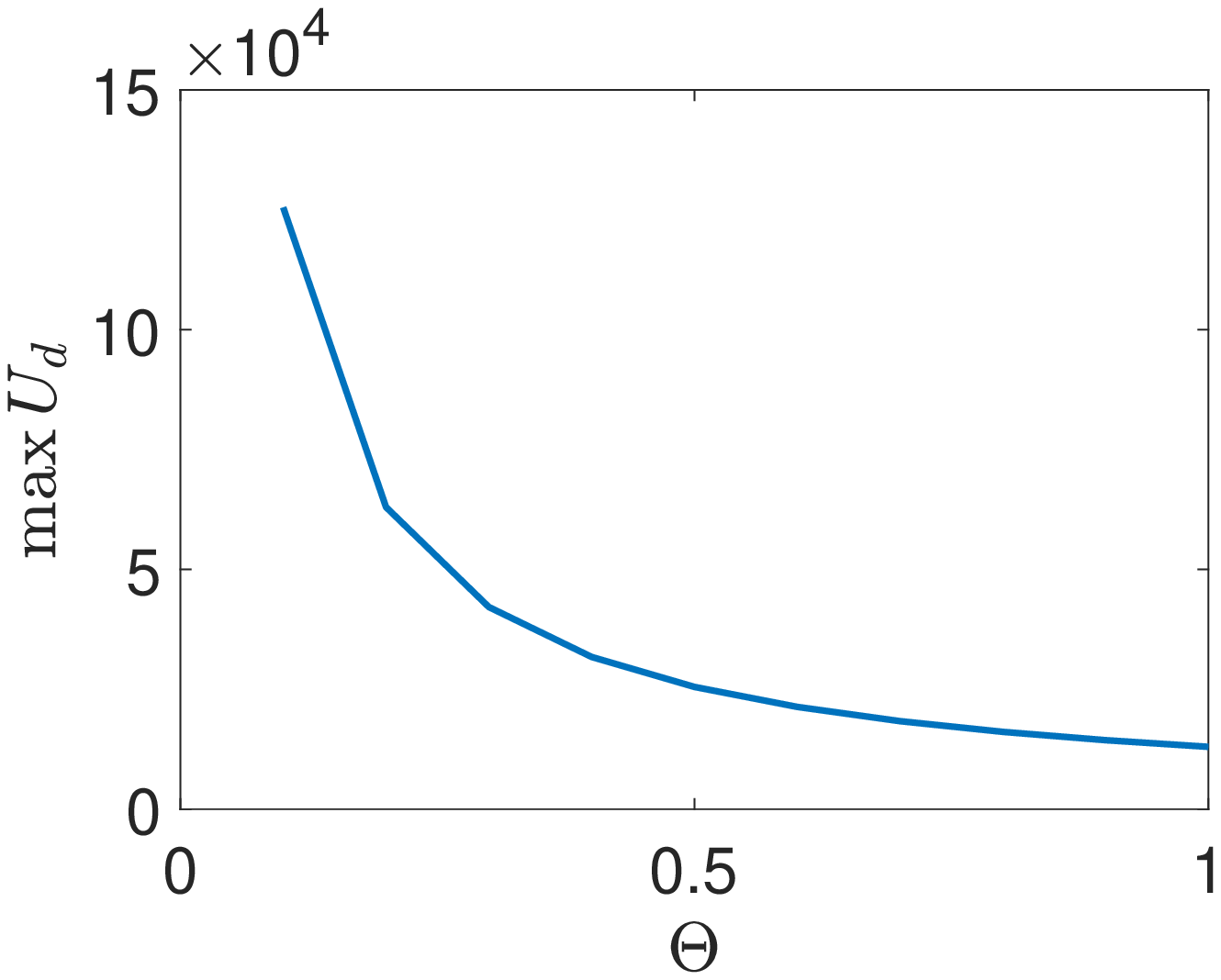}}
\subfigure{
\includegraphics[width=0.23\textwidth]{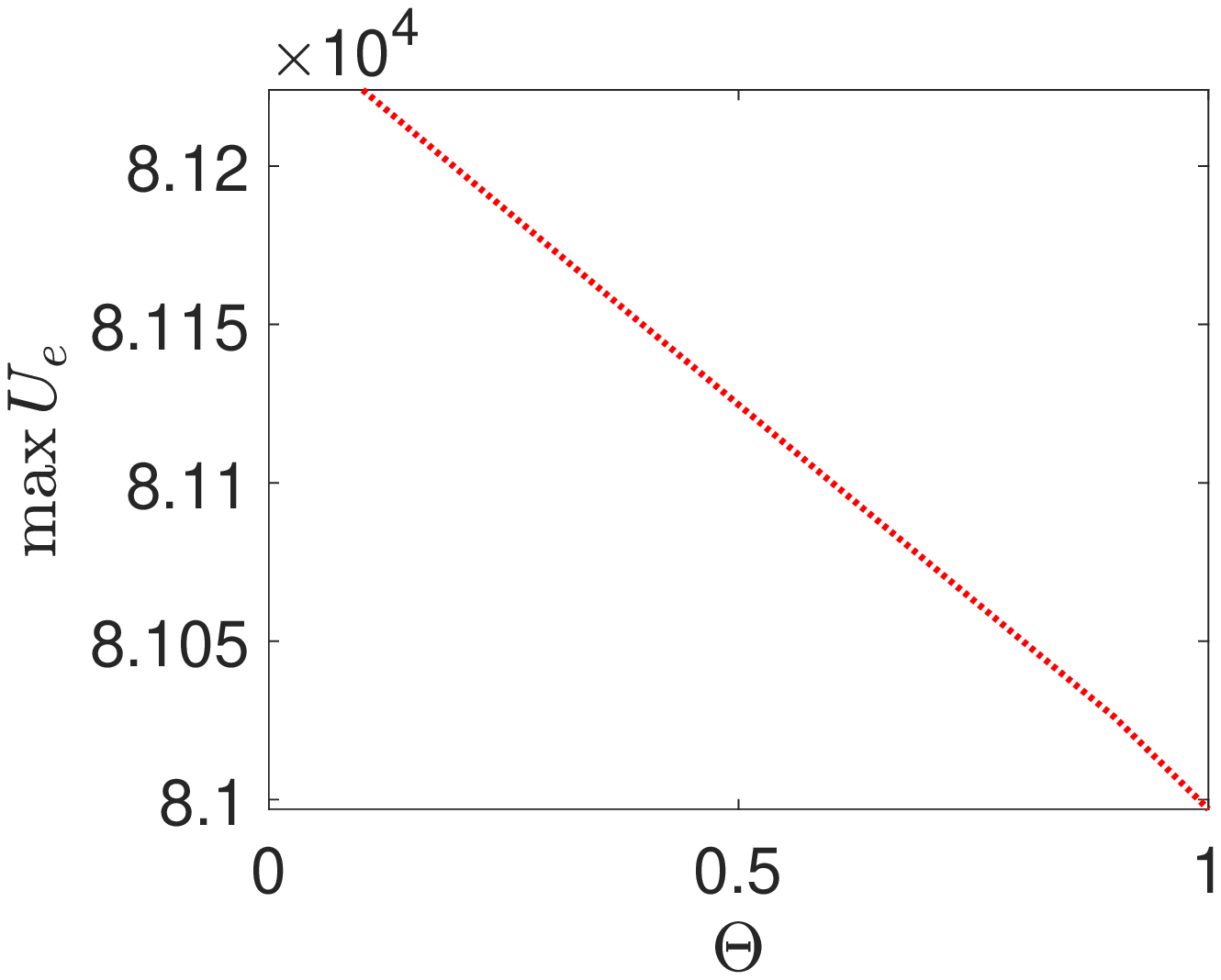}}
\caption{Maximized utility changing with $\Theta$.}
\label{fig:utility_theta}
\end{figure}

\begin{figure}
\subfigure{
\includegraphics[width=0.22\textwidth]{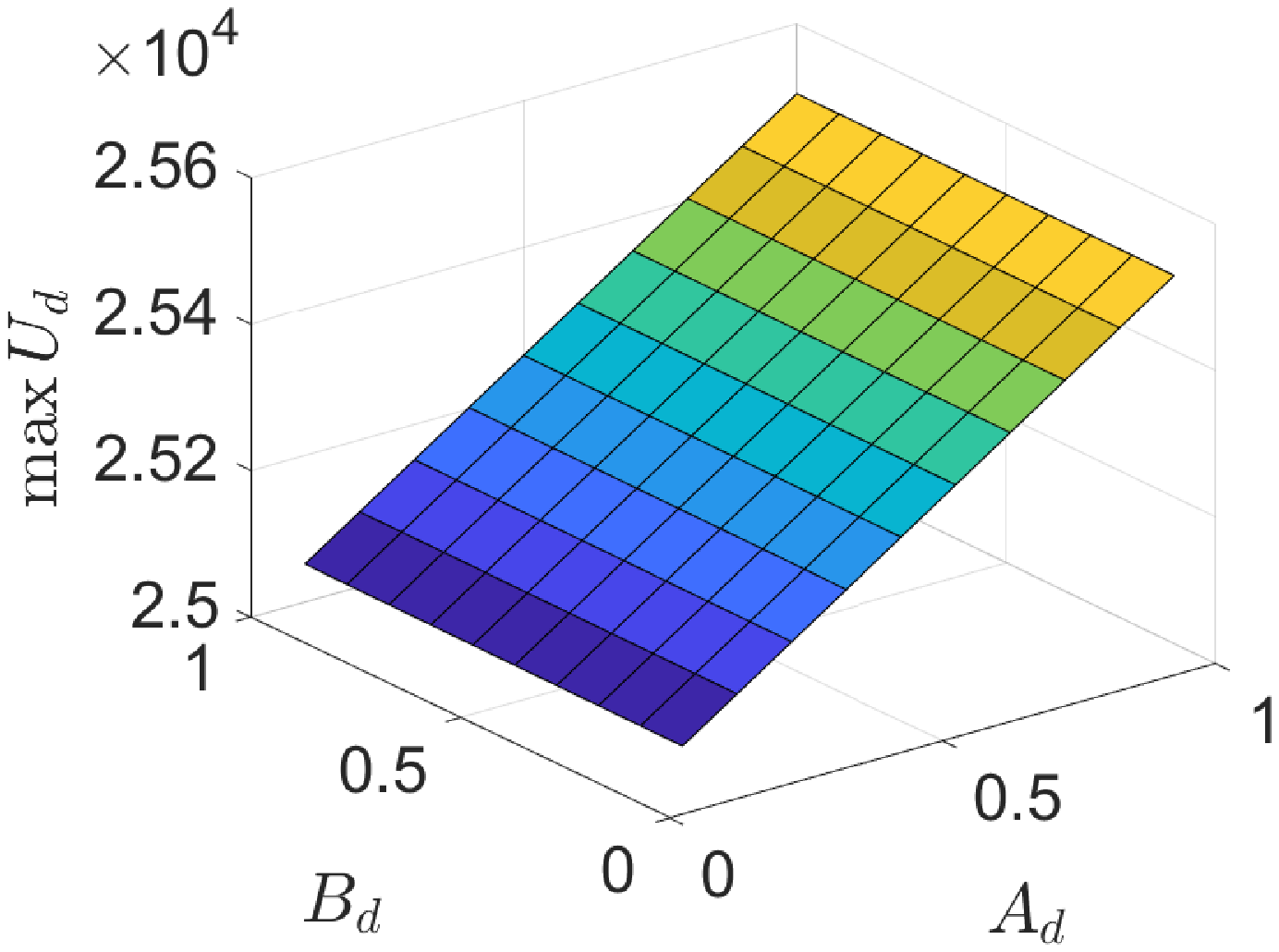}}
\subfigure{
\includegraphics[width=0.22\textwidth]{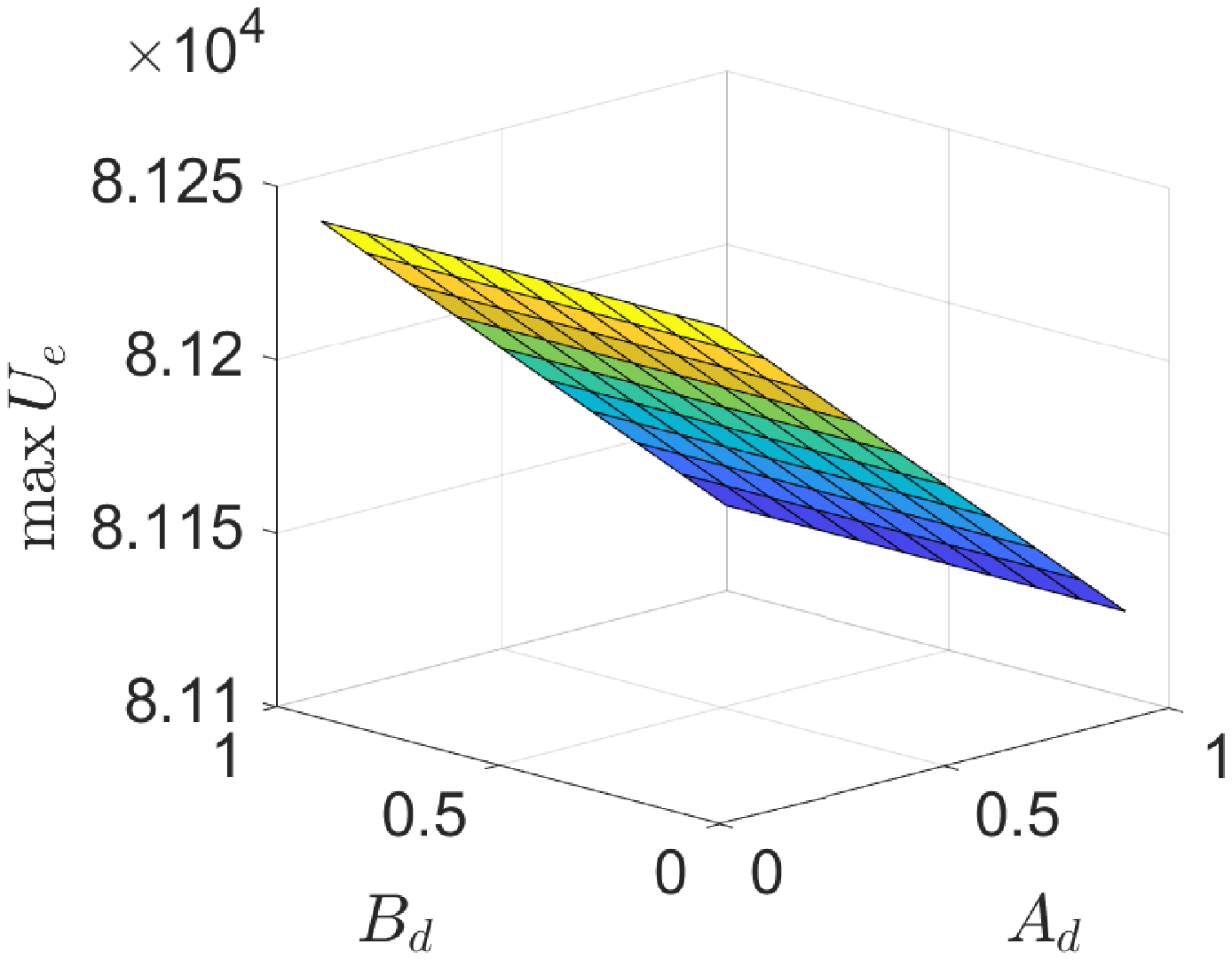}}
\caption{Maximized utility changing with $A_d$ and $B_d$.}
\label{fig:utility_Ad_Bd}
\end{figure}

Next, we evaluate how the server's related parameters affect their utilities. To begin with, we change $A_e$ and $B_e$ in the same range with the same interval as that of $A_d$ and $B_d$, which brings the utility results as reported in Fig. \ref{fig:utility_Ae_Be}. For the maximized utility of the device, the increase of $A_e$ makes it decrease in an inverse proportional fashion while the change of $B_e$ brings no difference to it. This is because the definition of $U_d$ is not related to $B_e$ but indirectly affected by $A_e$ due to the expressions of $r^*(s)$ and $s^*$. 
For the maximized utility of the server, $A_e$ brings no impact while the increasing $B_e$ decreases the optimal value of $U_e$, where the potential reason might be that the production of $s$ and $r$ in $y(\Theta,s)$ makes the activity of $A_e$ offsetting but $B_e$ still functions negatively for $U_e$.

\begin{figure}
\subfigure{
\includegraphics[width=0.22\textwidth]{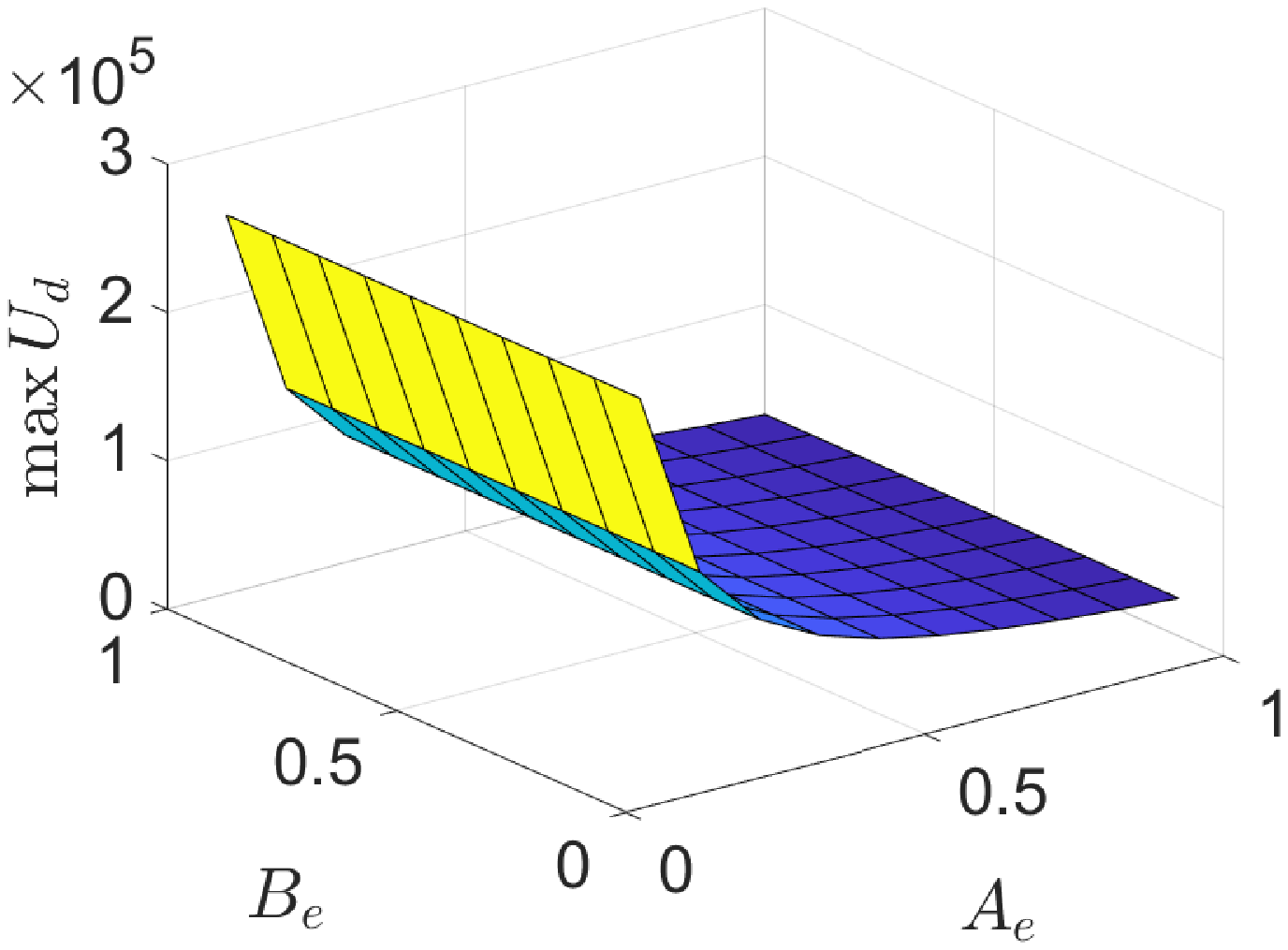}}
\subfigure{
\includegraphics[width=0.22\textwidth]{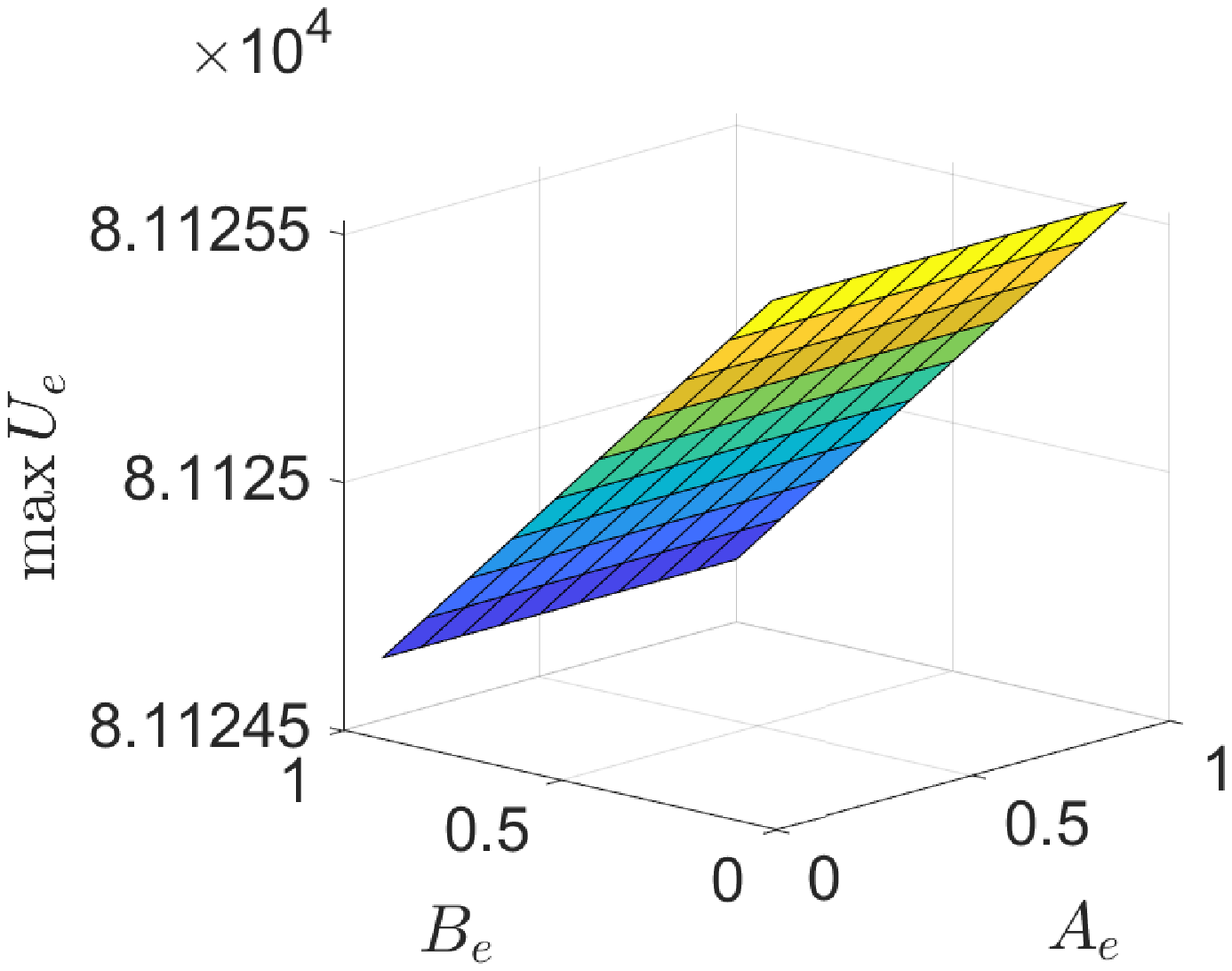}}
\caption{Maximized utility changing with $A_e$ and $B_e$.}
\label{fig:utility_Ae_Be}
\end{figure}

Then the impacts of $\sigma$ and $\rho$ are reported in Fig. \ref{fig:utility_sigma_rho} while the experimental results of changing $s_0$ and $r_0$ are demonstrated in Fig. \ref{fig:utility_s0_r0}. From Fig. \ref{fig:utility_sigma_rho}, we can see, on one hand, $\sigma$ does not affect the value of maximized $U_d$ while $\rho$ presents negligible influence on the value of maximized $U_e$; on the other hand, the maximized $U_d$ increases with $\rho$ and the maximized $U_e$ arises with $\sigma$. 
From Fig. \ref{fig:utility_s0_r0}, it is clear that the change of $s_0$ has no impact on both the maximized $U_d$ and $U_e$, while the increase of $r_0$ makes $U_d$ increase but $U_e$ decrease. 
All these phenomena are intelligible according to the expressions of $U_d$ and $U_e$, as well as $r^*(s)$ and $s^*$, mentioned in Section \ref{sec:mechanism}.

\begin{figure}
\subfigure{
\includegraphics[width=0.22\textwidth]{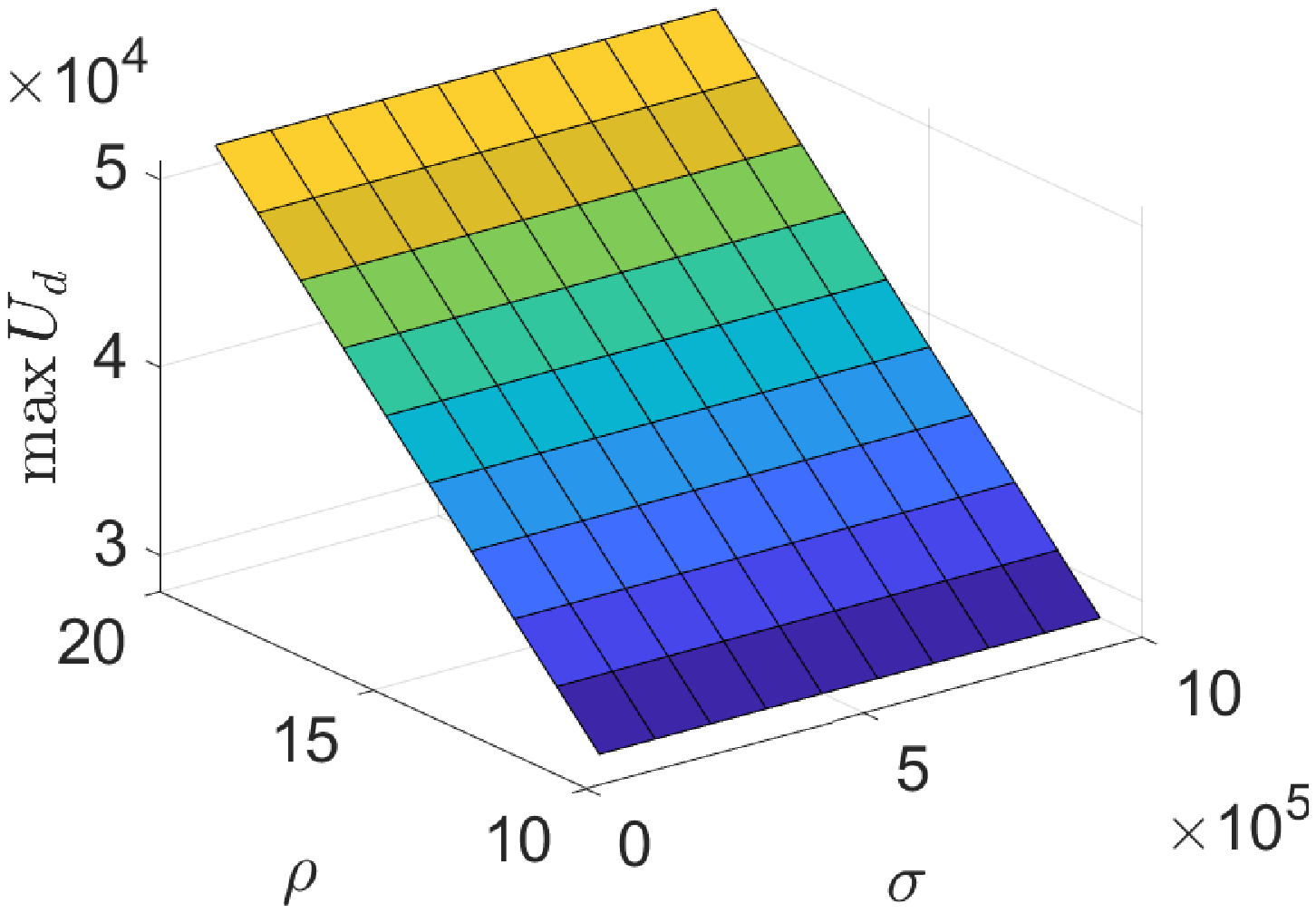}}
\subfigure{
\includegraphics[width=0.22\textwidth]{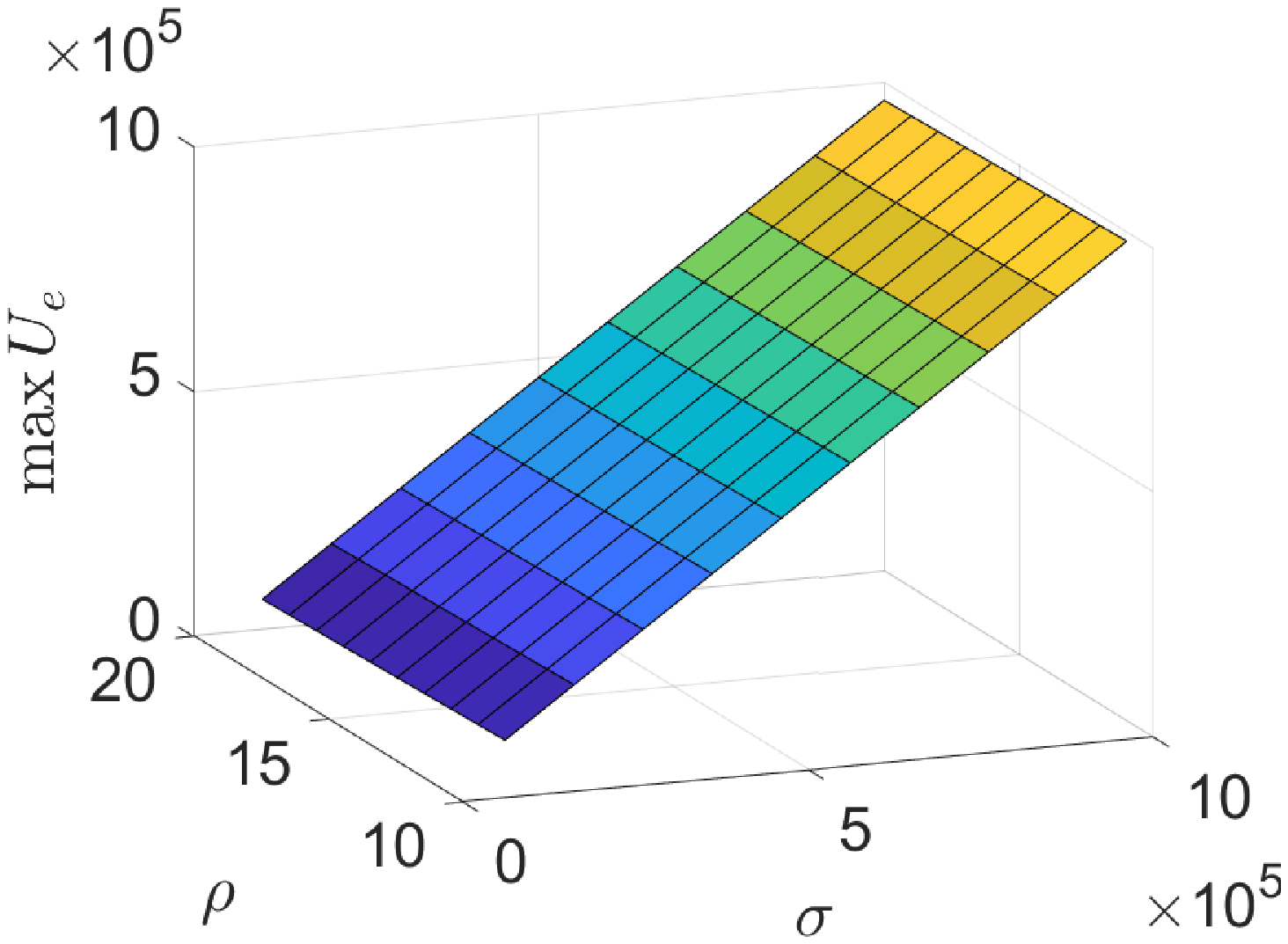}}
\caption{Maximized utility changing with $\sigma$ and $\rho$.}
\label{fig:utility_sigma_rho}
\end{figure}

\begin{figure}
\subfigure{
\includegraphics[width=0.22\textwidth]{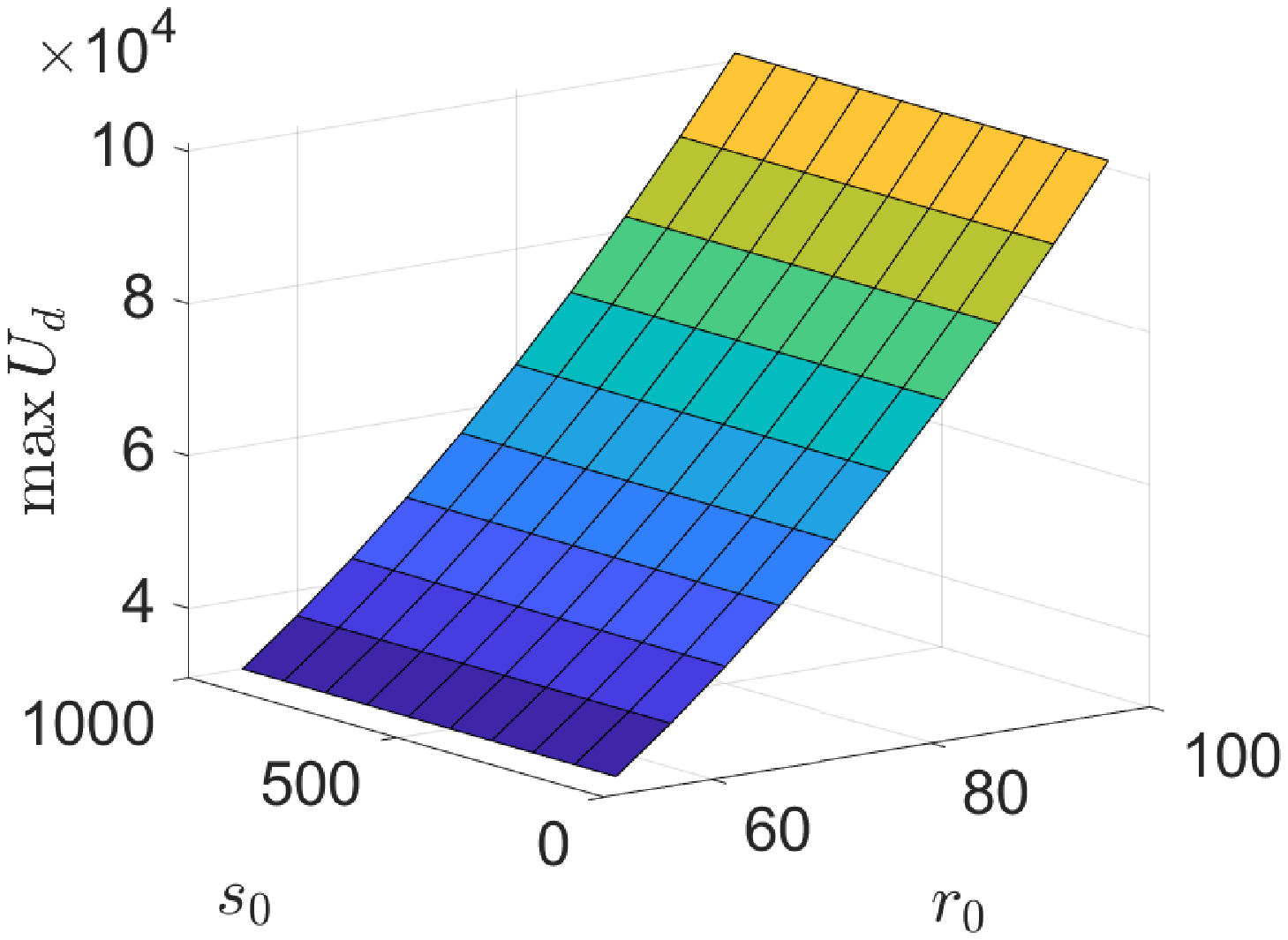}}
\subfigure{
\includegraphics[width=0.22\textwidth]{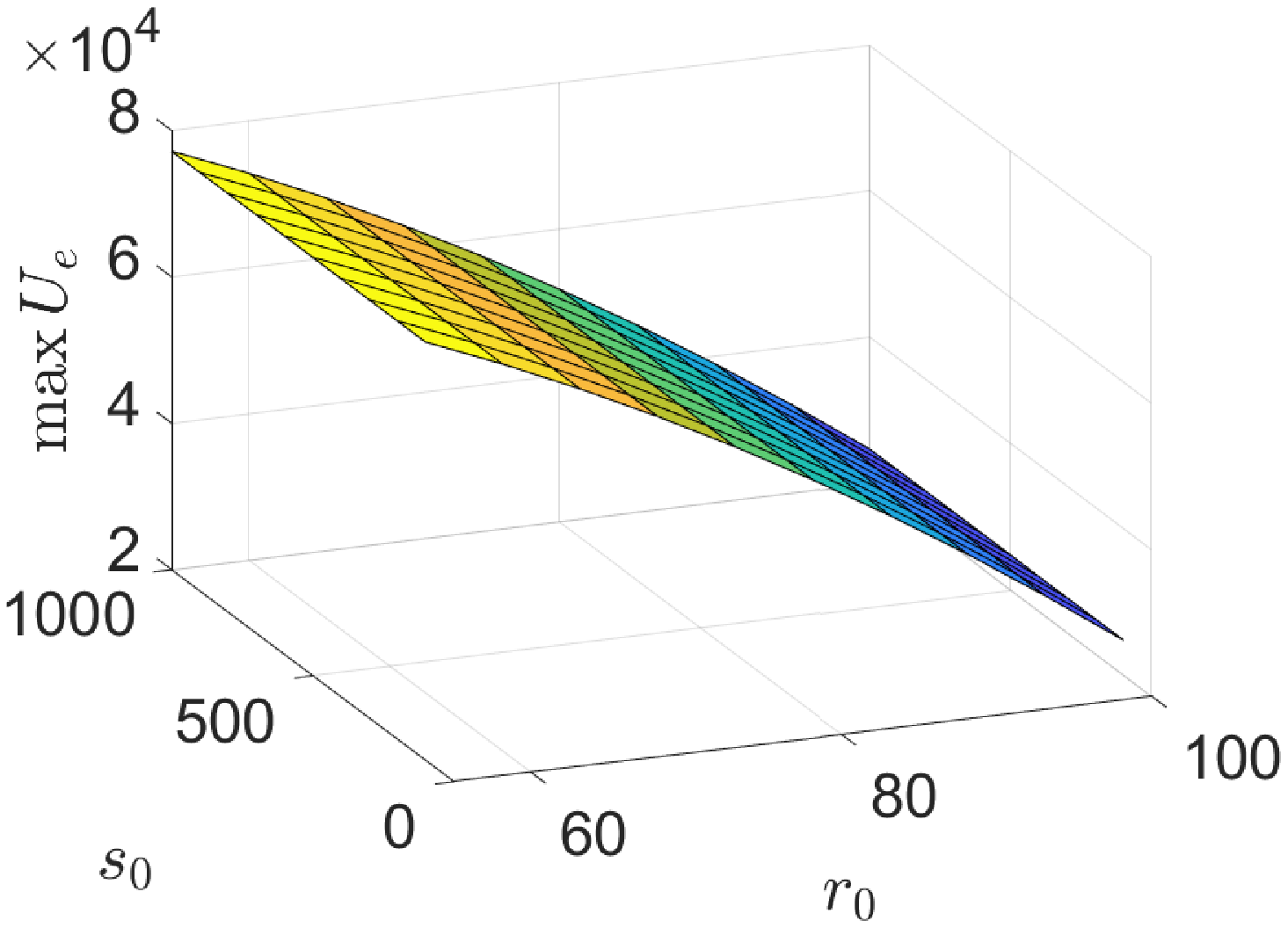}}
\caption{Maximized utility changing with $s_0$ and $r_0$.}
\label{fig:utility_s0_r0}
\end{figure}

\subsection{Participation Decision Making}\label{subsec:exp_participation}
We then evaluate the performance of the proposed participation decision making scheme in Section \ref{sec:solution}. Relevant scalar parameters are set as 
$\alpha=10,\beta_i=10^{-3},\gamma_i=10^{-5}, \delta = 10^{-3}, ~\mathrm{and}~ w=3.5\times 10^5$. Other sets of parameters are also investigated, which produce similar results and thus are omitted. 

To begin with, we simulate the improved solution proposed in Section \ref{subsec:improve} and compare it with the direct solution in Section \ref{subsec:game} in terms of both the time complexity and optimization objective. 
To implement this series of experiments, we first derive error-related parameters $a,b$ defined in \eqref{eq:total} using the MNIST dataset \cite{lecun1998gradient} with up to 6,000 samples to train a 2-layer CNN classifier, generating the actual error results in Fig. \ref{fig:error_time}(a) which is fitted by the power-law function with $a=13.2,b=0.7$ (95\% confidence). 

Next, we set the data size of any device with 50 or 500 randomly, change the number of devices $n$ from 2 to 16, and run both the direct and improved solutions, where the improved solution is implemented with the number of small device sets $\xi = 2$. 
The comparison results of computational cost is presented in Fig. \ref{fig:error_time}(b) and the maximized total profit is reported in Table \ref{tab:profit}. One can find that the improved solution can bring approximate results as the direct one in terms of maximizing the global profit but consume much less time when $n$ is larger, which indicates that the proposed improved method can function effectively.

\begin{figure}
\subfigure[Error]{
\includegraphics[width=0.22\textwidth]{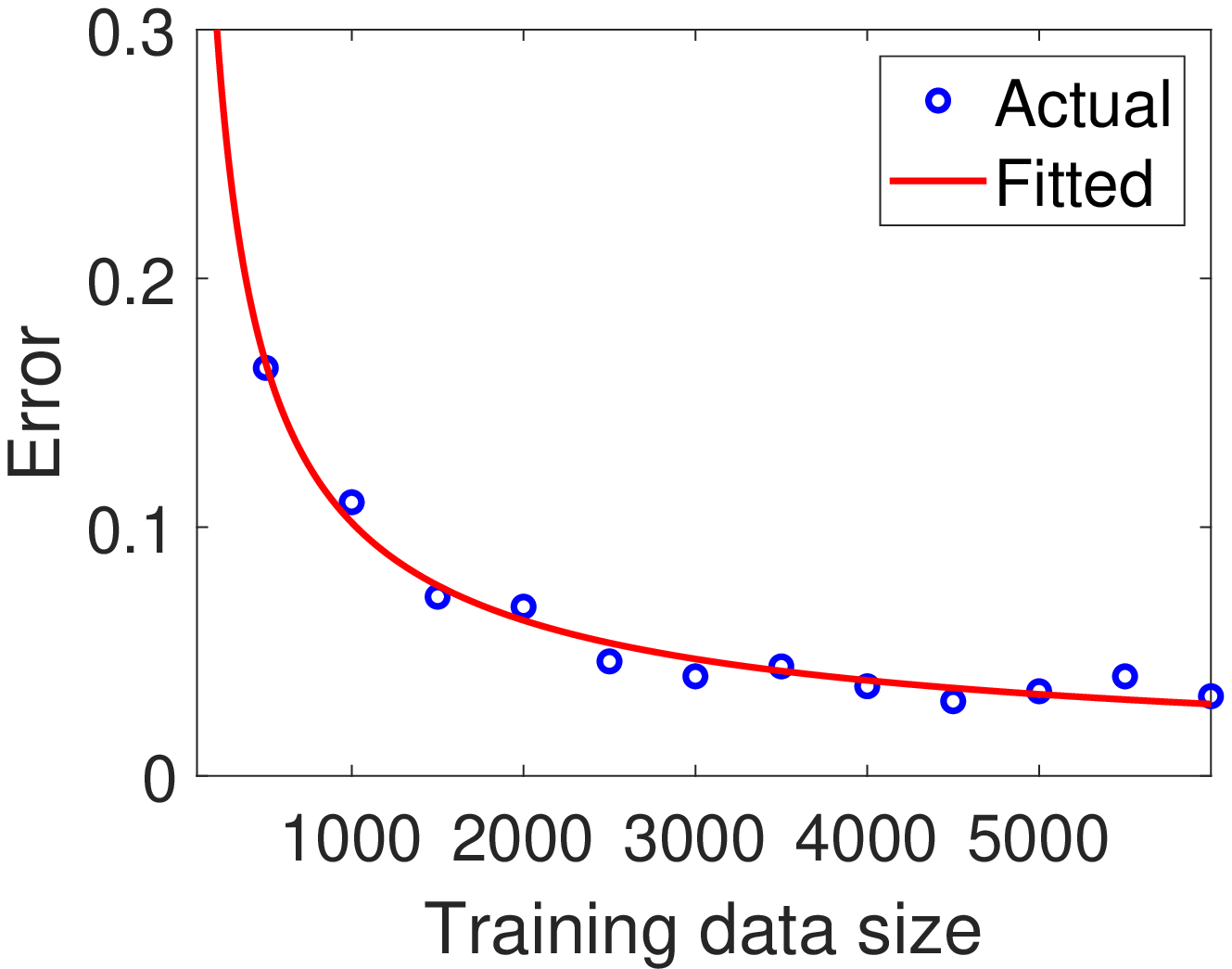}
}
\subfigure[Computational cost]{
\includegraphics[width=0.22\textwidth]{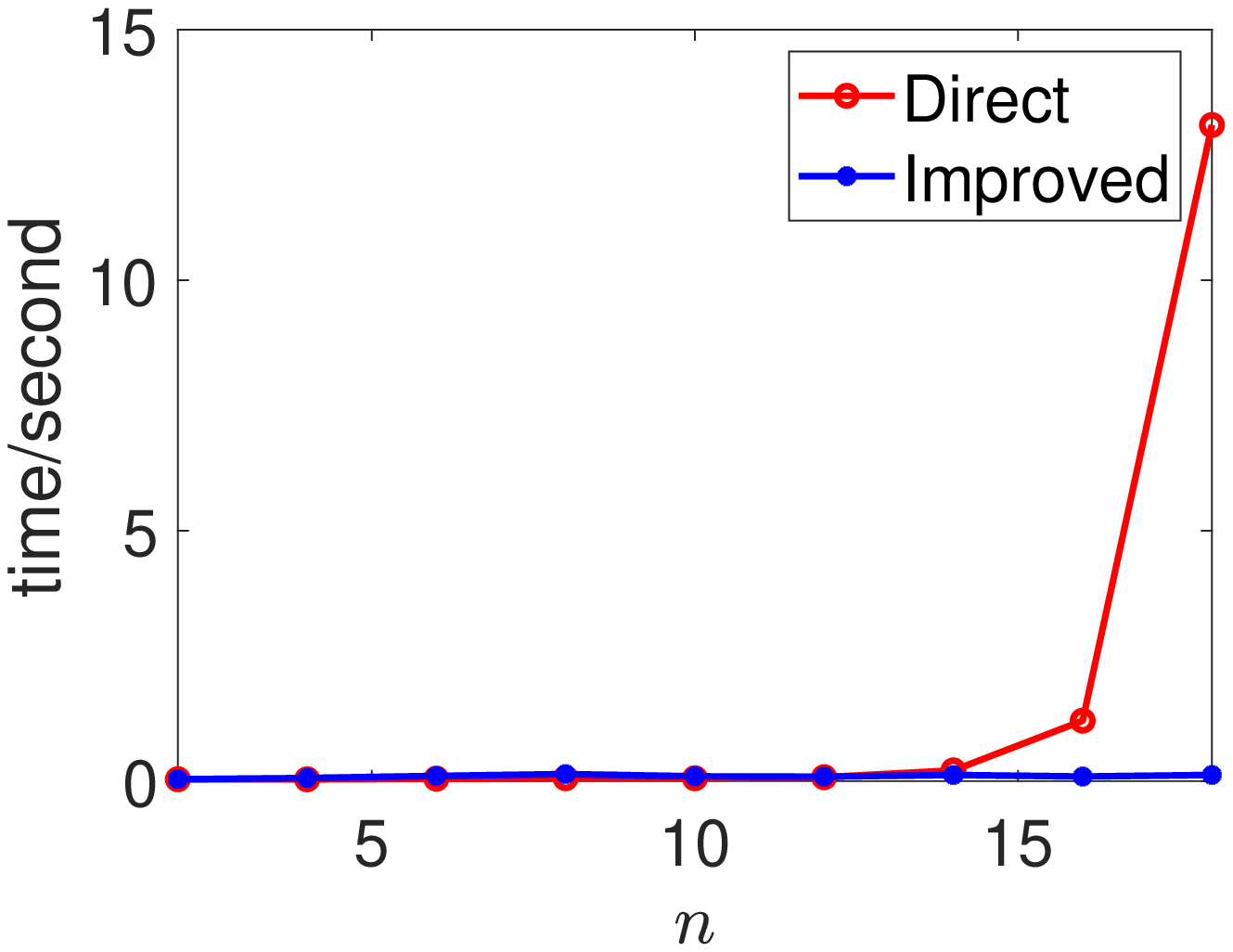}
}
\caption{Error fitting and computational cost comparison.}
\label{fig:error_time}
\end{figure}

\begin{table}
\caption{Comparison of maximized total profit.}
\centering
\begin{tabular}{ |c||c|c|c|c|c|c| } 
\hline
Solution & $n$=2 & $n$=4 & $n$=6 & $n$=8  & $n$=10 & $n$=12\\
\hline
Direct & 0.95 &	2.60 &	2.84 &	2.94  & 2.99 & 3.02\\
\hline
Improved & 0.95 & 2.62 & 2.60 &	3.05 & 2.61 & 2.77 \\
\hline
\end{tabular}
\label{tab:profit}
\end{table}

Besides, we change the number of small device sets in the improved solution as $\xi\in \{2,3,4,5\}$ and study its impact on the computational cost. As shown in Fig. \ref{fig:time_setnumber}, we investigate four scenarios with different number of devices $n\in \{10,15,20,30\}$. For clear presentation, we report the results for $n = 10,15,20$ and $n = 10,20,30$, separately. 
From Fig. \ref{fig:time_setnumber}(a), we can see that for a given $n$, the larger the number of small device sets $\xi$, the higher the computation cost of the improved solution. This is because we simulate the operation of solving SGPM problems in a serialized manner where each SGPM takes some time to finish; while the number of devices in each small set seldom affects the running time for these given $n$. However, when we increase $n$ to a larger value, such as $n=30$ in the right-side Fig. \ref{fig:time_setnumber}(b), the running cost for $\xi=2$ will be far larger than all other cases. This is because every SGPM problem in the improved solution for $n=30$ and $\xi=2$ is corresponding to the GPM problem in the direction solution for $n=15$, costing much more time than smaller $n$, which can also been observed from the results in Fig. \ref{fig:error_time}(b). Despite this case, other results are still consistent with the feature of larger time cost for larger $\xi$.

\begin{figure}
\subfigure{
\includegraphics[width=0.22\textwidth]{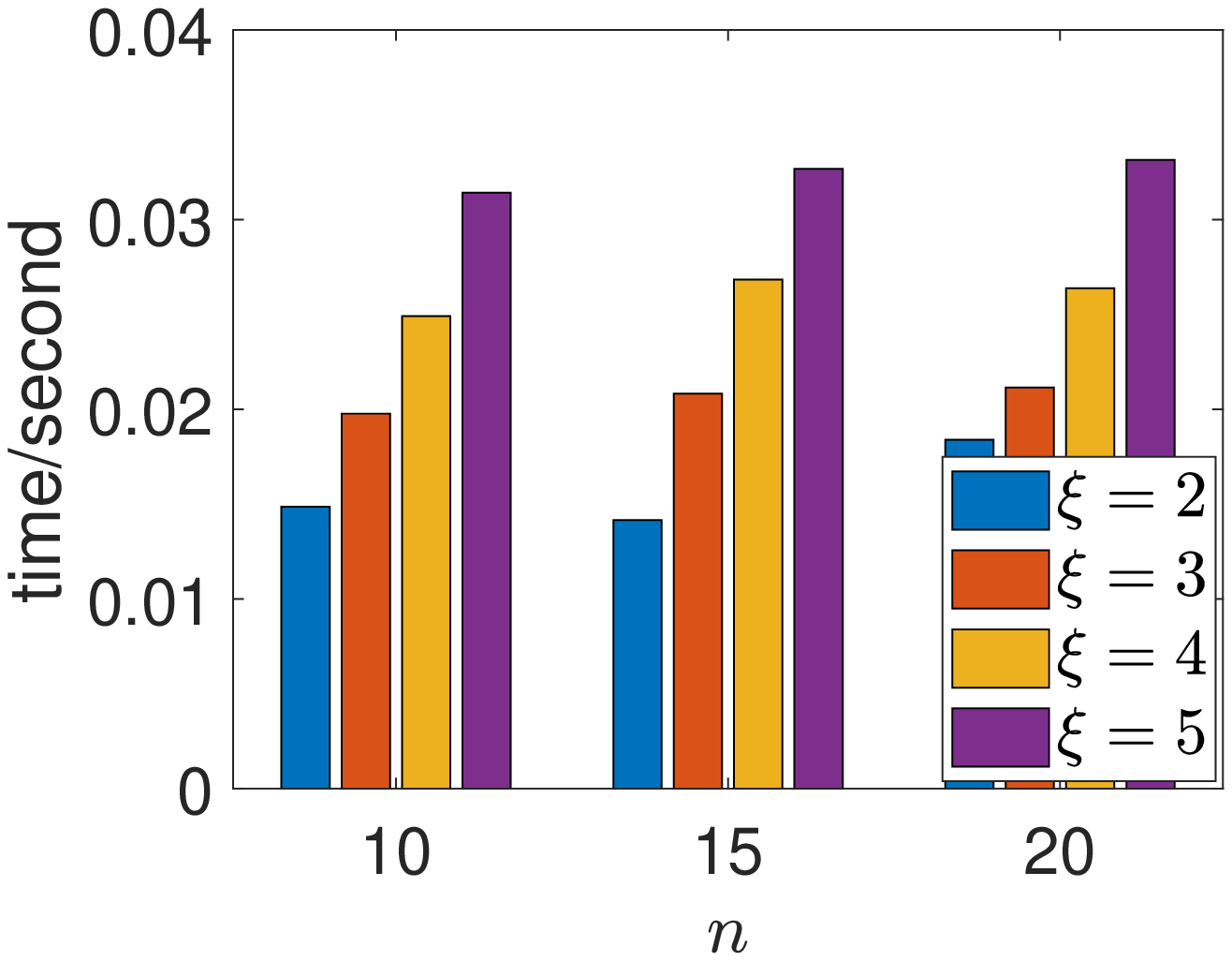}
}
\subfigure{
\includegraphics[width=0.22\textwidth]{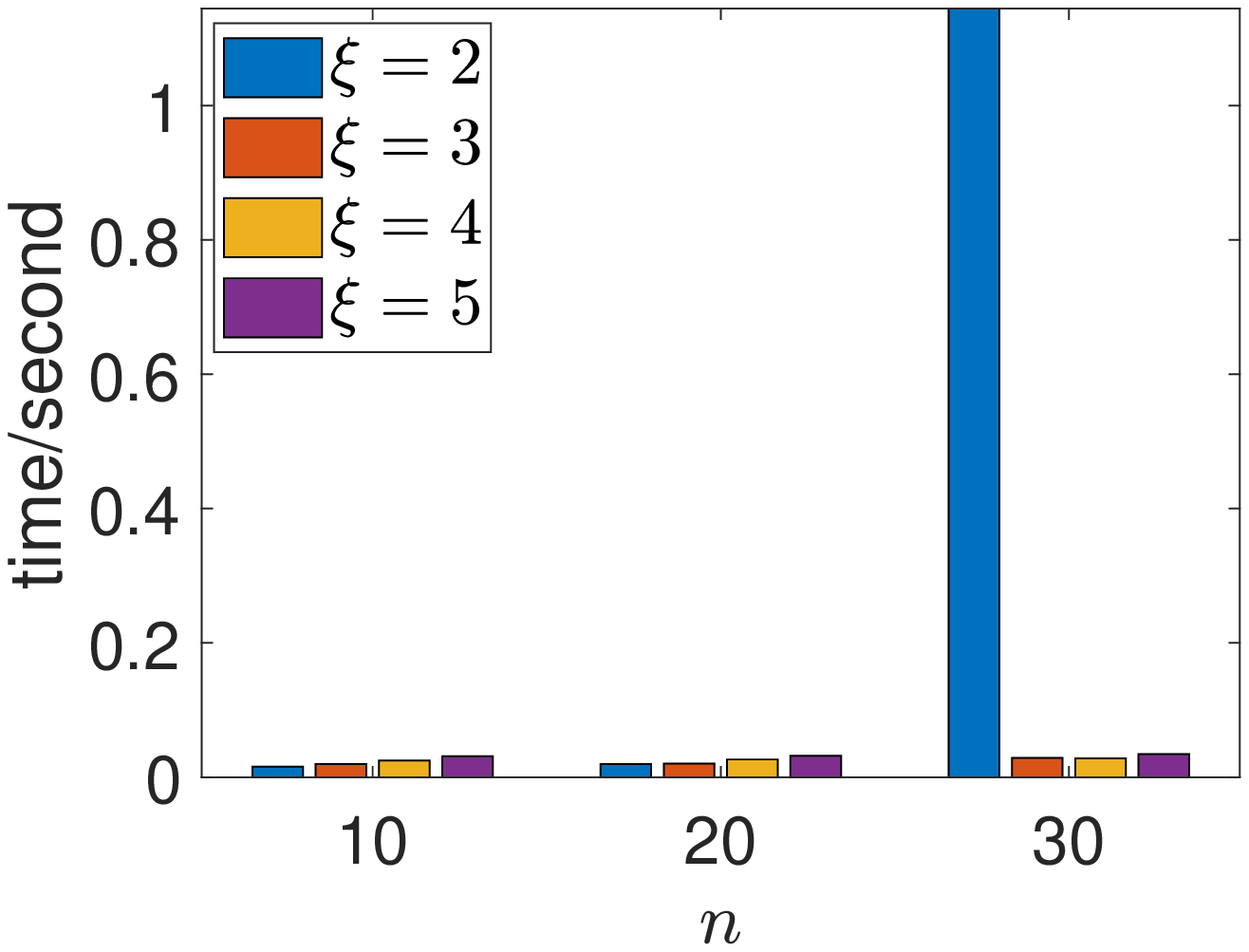}
}
\caption{Computational cost of the improved solution with different number of devices $n$ and number of small sets $\xi$. }
\label{fig:time_setnumber}
\end{figure}

Further, we investigate the impact of data size $s_i$ on both the maximized total profit and participation decision results. In detail, we take $n=8$ as an example and select $d_1$ as a test sample with variable size $s_1\in \{50,100,200,400,600,800,1000\}$. To broadly study this problem, we consider two representative cases where all other devices have the same size of generated data, denoted by $\mathbf{s}_{-1}$, as 50 and 500, and the experimental results of two cases are reported in Figs. \ref{fig:case50} and \ref{fig:case500}, respectively. From Figs. \ref{fig:case50}(a) and \ref{fig:case500}(a), we can see that in both cases, the maximized total profit of all devices increases with $s_1$, revealing that the more data contributed to the edge server, the larger the total profit for all. This is reasonable since more data can facilitate training a better ML model to benefit all devices. While from Figs. \ref{fig:case50}(b) and \ref{fig:case500}(b), one can see that the participation decision vectors in two cases differ a lot. In particular, when $\mathbf{s}_{-1}=50$, $p_1$ becomes 1 earlier at $s_1=100$; but in the case of $\mathbf{s}_{-1}=500$, only when $s_1$ reaches 600, $p_1=1$ happens. This phenomenon is understandable as only when the data contribution of $d_1$ is large enough compared to others, should it take part in the FEL to help train a better model; otherwise, there will only be more resource consumption but nothing contributed to the model training.

\begin{figure}[h]
\subfigure[Total profit]{
\includegraphics[width=0.23\textwidth]{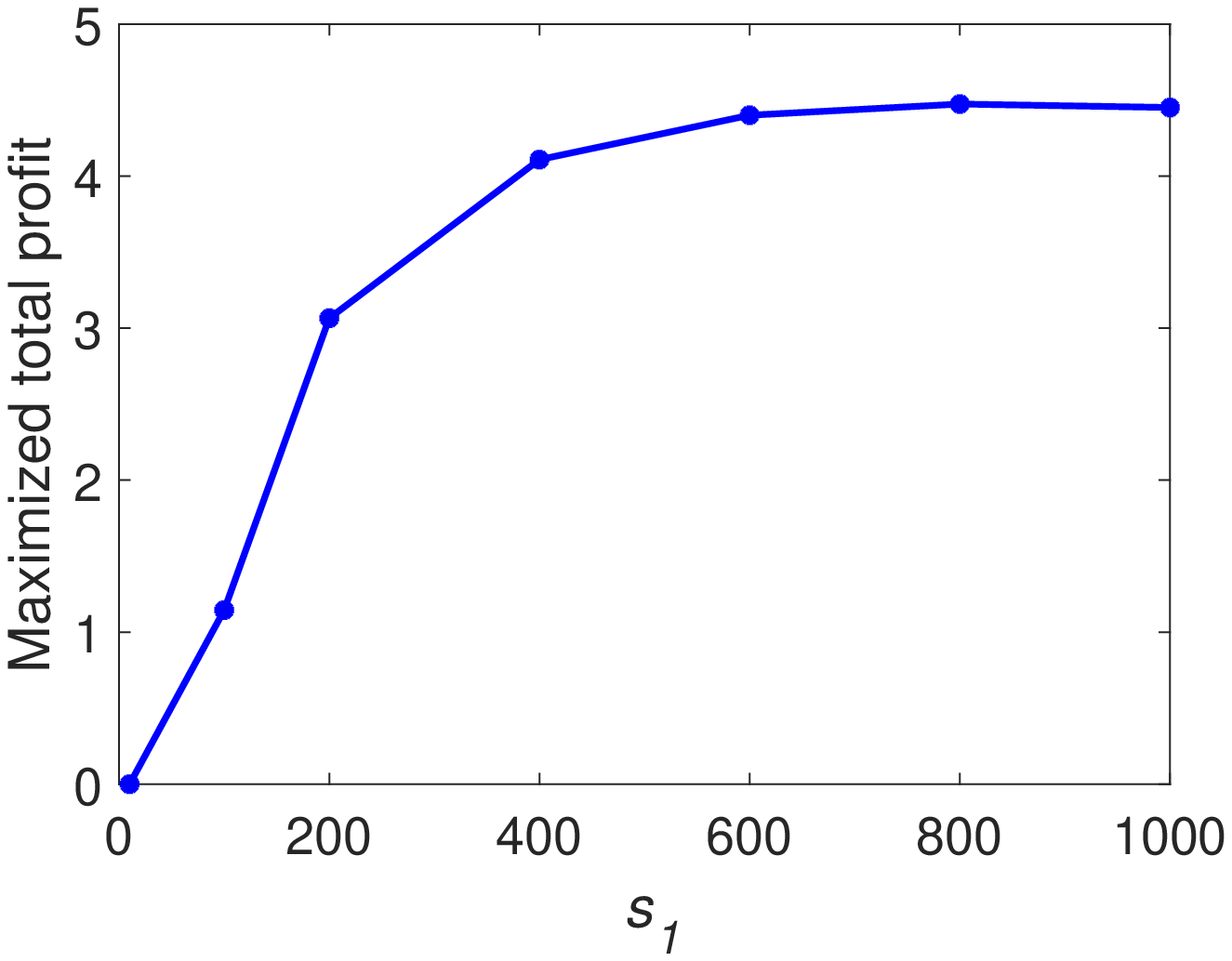}}
\subfigure[Participation decision]{
\includegraphics[width=0.21\textwidth]{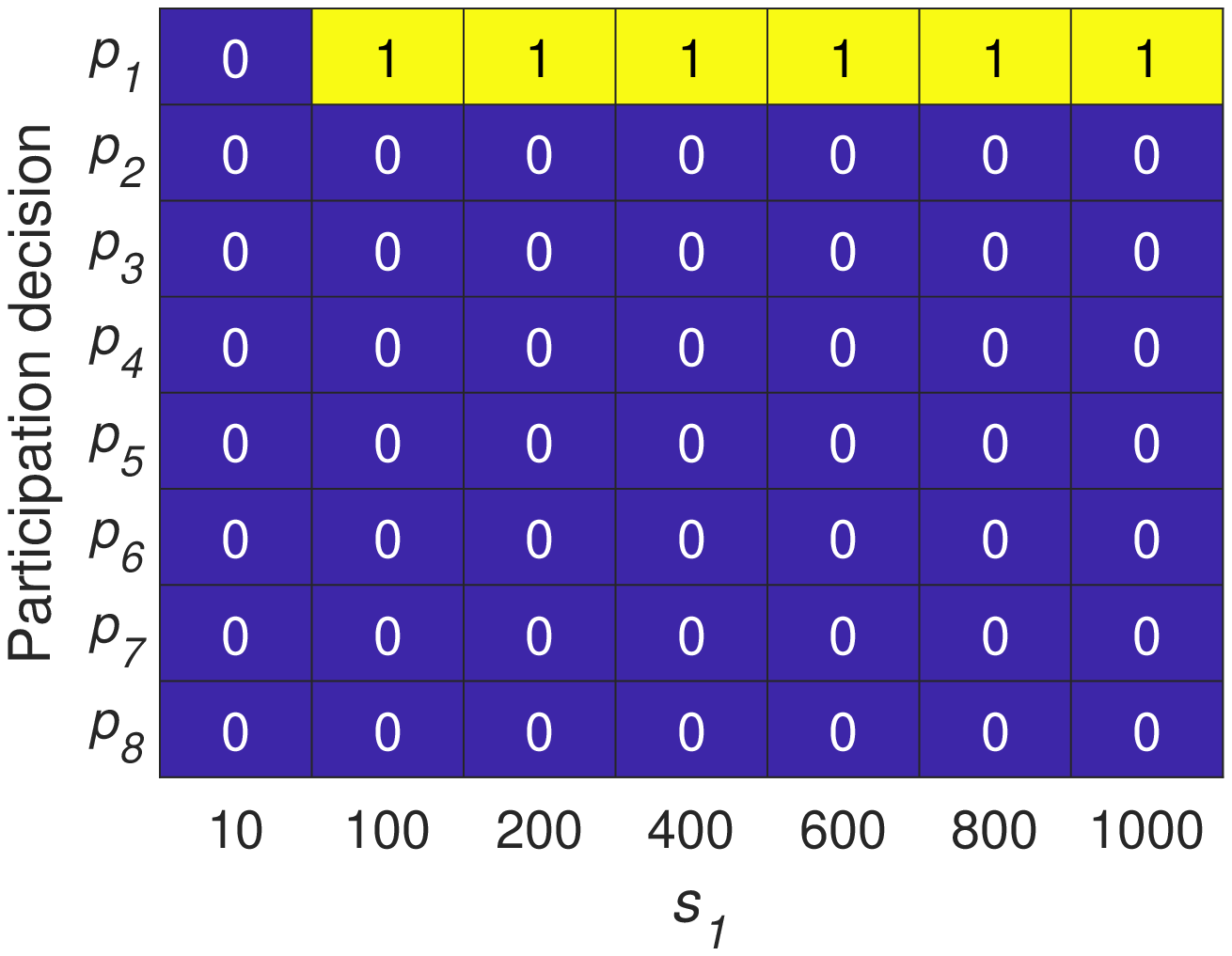}}
\caption{Case of $\mathbf{s}_{-1}=50$.}
\label{fig:case50}
\end{figure}

\begin{figure}[h]
\subfigure[Total profit]{
\includegraphics[width=0.23\textwidth]{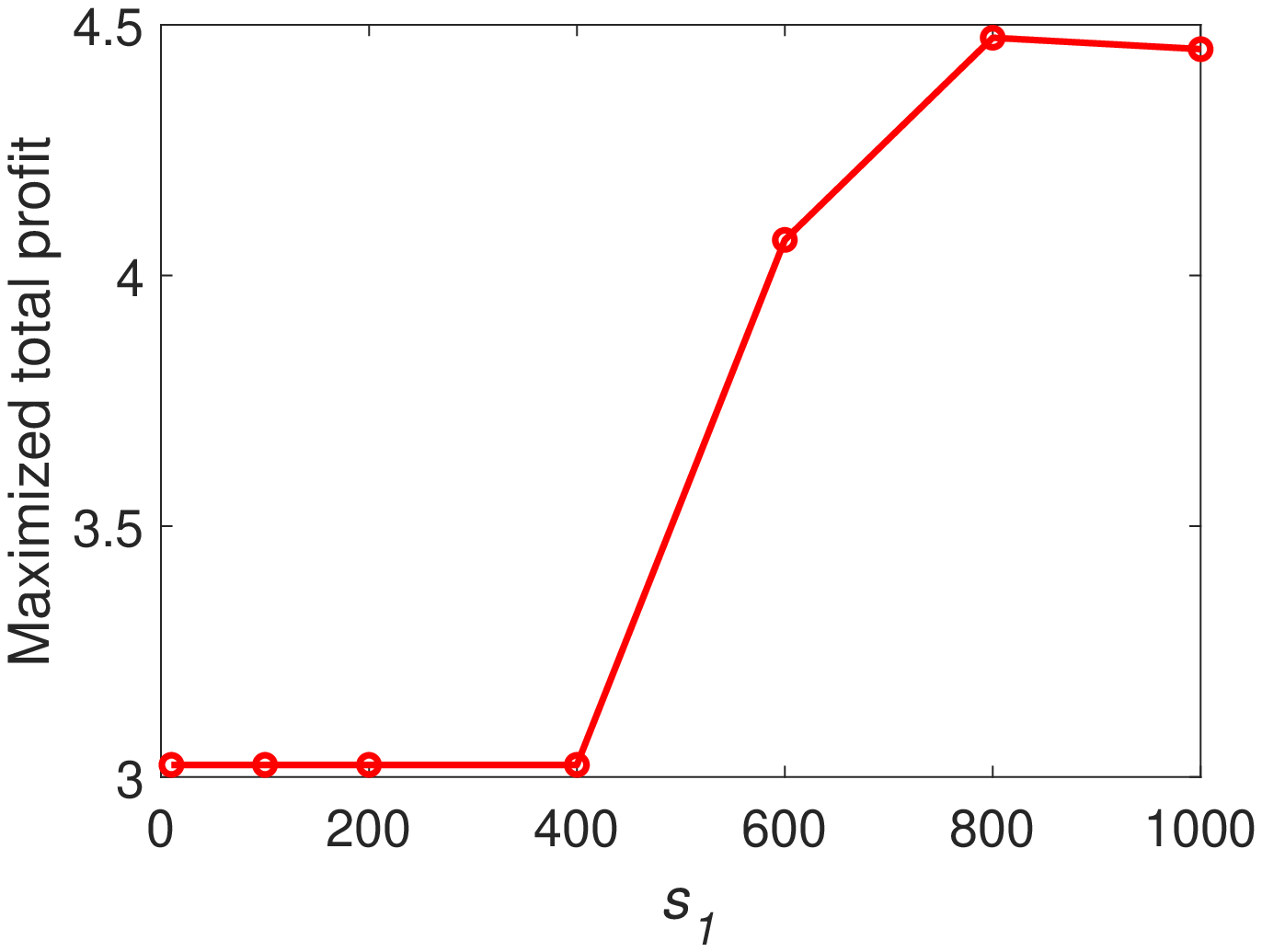}}
\subfigure[Participation decision]{
\includegraphics[width=0.21\textwidth]{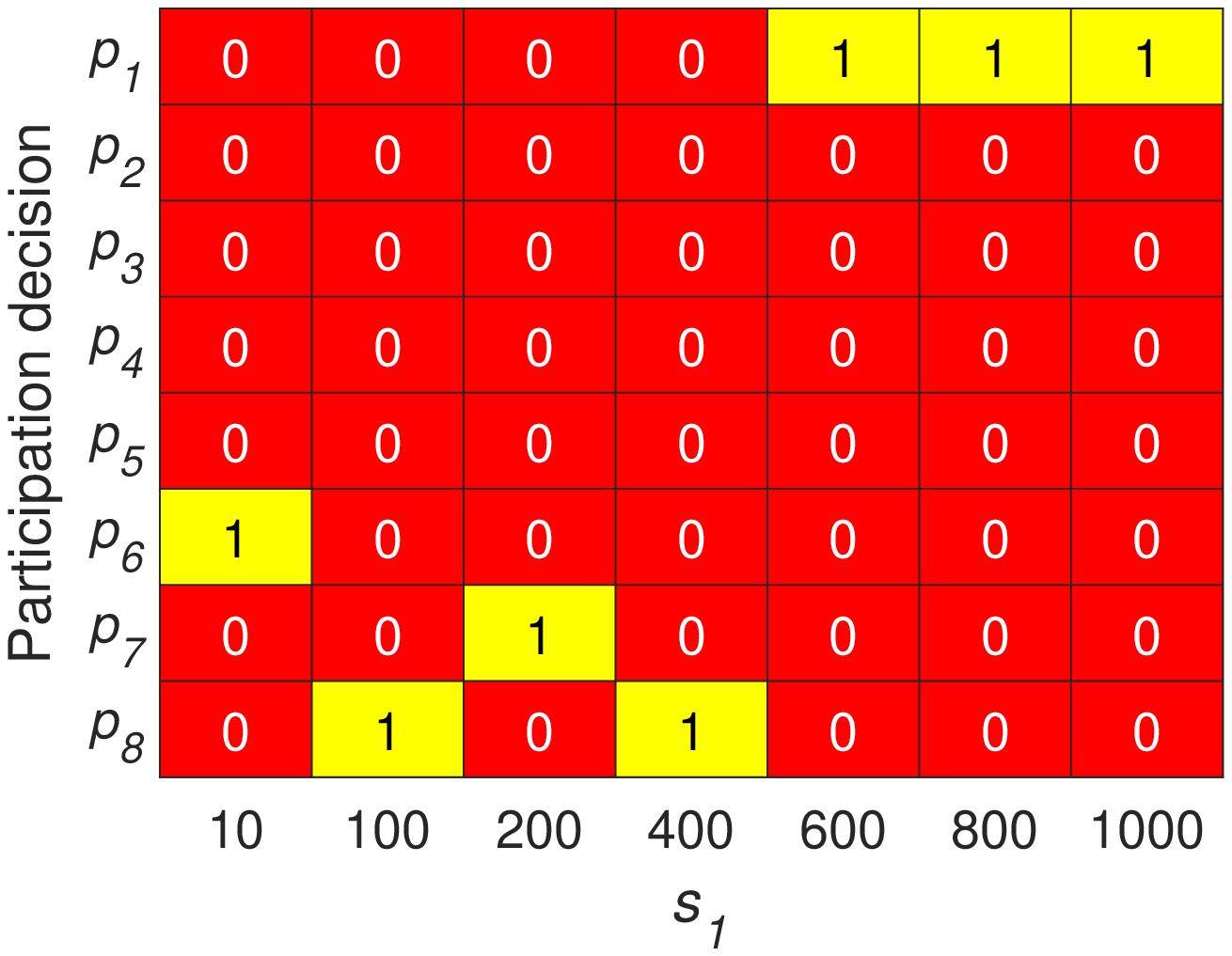}}
\caption{Case of $\mathbf{s}_{-1}=500$.}
\label{fig:case500}
\end{figure}

\section{Conclusion and Future Work}\label{sec:conclusion}

In this paper, we study the issue of FEL system composition considering more about the interests of edge devices with the aim of maintaining the efficiency and sustainability in the long term. Different from most of the existing studies on FEL performance improvement focusing on the optimization and control during the learning process, we take precautions to establish the best organization for FEL. Specifically, we first quantify the impact of local training data sizes to define a participation game sketching the relationships among all devices. 
And a mechanism design based truthful data size collection process is elaborated to prepare for the design and operation of the game-theoretic decision scheme.  
Then the correlated equilibrium is introduced to guarantee the individual optimum, facilitating the game-theoretic solution for participation decision making. An improved method is further proposed to reduce the computational complexity to polynomial time. 
Finally, both proposed schemes are evaluated with real-world data and simulation experiments.

For the future research, we will investigate the case of  lacking mutual trust between the edge server and devices in FEL, where the server might also behave maliciously for intentionally exploiting their contributions in learning. Besides, as mobile devices massively participate in FEL, their mobility results in more challenges for guaranteeing the learning performance, where cross-edge scenarios will be thoroughly explored.


\bibliographystyle{IEEEtran}
\bibliography{reference}

\begin{thebibliography}{10}
\providecommand{\url}[1]{#1}
\csname url@samestyle\endcsname
\providecommand{\newblock}{\relax}
\providecommand{\bibinfo}[2]{#2}
\providecommand{\BIBentrySTDinterwordspacing}{\spaceskip=0pt\relax}
\providecommand{\BIBentryALTinterwordstretchfactor}{4}
\providecommand{\BIBentryALTinterwordspacing}{\spaceskip=\fontdimen2\font plus
\BIBentryALTinterwordstretchfactor\fontdimen3\font minus
  \fontdimen4\font\relax}
\providecommand{\BIBforeignlanguage}[2]{{%
\expandafter\ifx\csname l@#1\endcsname\relax
\typeout{** WARNING: IEEEtran.bst: No hyphenation pattern has been}%
\typeout{** loaded for the language `#1'. Using the pattern for}%
\typeout{** the default language instead.}%
\else
\language=\csname l@#1\endcsname
\fi
#2}}
\providecommand{\BIBdecl}{\relax}
\BIBdecl

\bibitem{xiao2019edge}
Y.~Xiao, Y.~Jia, C.~Liu, X.~Cheng, J.~Yu, and W.~Lv, ``Edge computing security:
  State of the art and challenges,'' \emph{Proceedings of the IEEE}, vol. 107,
  no.~8, pp. 1608--1631, 2019.

\bibitem{market}
``Edge computing market,'' \url{https://www.marketsandmarkets.com
  /Market-Reports/edge-computing-market-133384090.html}, accessed: 2020-11-30.

\bibitem{yang2020artificial}
H.~Yang, A.~Alphones, Z.~Xiong, D.~Niyato, J.~Zhao, and K.~Wu,
  ``Artificial-intelligence-enabled intelligent 6g networks,'' \emph{IEEE
  Network}, vol.~34, no.~6, pp. 272--280, 2020.

\bibitem{abad2020hierarchical}
M.~S.~H. Abad, E.~Ozfatura, D.~Gunduz, and O.~Ercetin, ``Hierarchical federated
  learning across heterogeneous cellular networks,'' in \emph{2020 IEEE
  International Conference on Acoustics, Speech and Signal Processing
  (ICASSP)}.\hskip 1em plus 0.5em minus 0.4em\relax IEEE, 2020, pp. 8866--8870.

\bibitem{zeng2019energy}
Q.~Zeng, Y.~Du, K.~K. Leung, and K.~Huang, ``Energy-efficient radio resource
  allocation for federated edge learning,'' \emph{arXiv preprint
  arXiv:1907.06040}, 2019.

\bibitem{yang2019scheduling}
H.~H. Yang, Z.~Liu, T.~Q. Quek, and H.~V. Poor, ``Scheduling policies for
  federated learning in wireless networks,'' \emph{IEEE Transactions on
  Communications}, 2020.

\bibitem{amiri2020update}
M.~M. Amiri, D.~Gunduz, S.~R. Kulkarni, and H.~V. Poor, ``Update aware device
  scheduling for federated learning at the wireless edge,'' in \emph{2020 IEEE
  International Symposium on Information Theory (ISIT)}, 2020.

\bibitem{yang2020age}
H.~H. Yang, A.~Arafa, T.~Q. Quek, and H.~V. Poor, ``Age-based scheduling policy
  for federated learning in mobile edge networks,'' in \emph{2020 IEEE
  International Conference on Acoustics, Speech and Signal Processing
  (ICASSP)}.\hskip 1em plus 0.5em minus 0.4em\relax IEEE, 2020, pp. 8743--8747.

\bibitem{lim2021decentralized}
W.~Y.~B. Lim, J.~S. Ng, Z.~Xiong, J.~Jin, Y.~Zhang, D.~Niyato, C.~Leung, and
  C.~Miao, ``Decentralized edge intelligence: A dynamic resource allocation
  framework for hierarchical federated learning,'' \emph{IEEE Transactions on
  Parallel and Distributed Systems}, vol.~33, no.~3, pp. 536--550, 2021.

\bibitem{lim2021dynamic}
W.~Y.~B. Lim, J.~S. Ng, Z.~Xiong, D.~Niyato, C.~Miao, and D.~I. Kim, ``Dynamic
  edge association and resource allocation in self-organizing hierarchical
  federated learning networks,'' \emph{IEEE Journal on Selected Areas in
  Communications}, vol.~39, no.~12, pp. 3640--3653, 2021.

\bibitem{wang2019adaptive}
S.~Wang, T.~Tuor, T.~Salonidis, K.~K. Leung, C.~Makaya, T.~He, and K.~Chan,
  ``Adaptive federated learning in resource constrained edge computing
  systems,'' \emph{IEEE Journal on Selected Areas in Communications}, vol.~37,
  no.~6, pp. 1205--1221, 2019.

\bibitem{zhu2019broadband}
G.~Zhu, Y.~Wang, and K.~Huang, ``Broadband analog aggregation for low-latency
  federated edge learning,'' \emph{IEEE Transactions on Wireless
  Communications}, vol.~19, no.~1, pp. 491--506, 2019.

\bibitem{yang2020federated}
K.~Yang, T.~Jiang, Y.~Shi, and Z.~Ding, ``Federated learning via over-the-air
  computation,'' \emph{IEEE Transactions on Wireless Communications}, vol.~19,
  no.~3, pp. 2022--2035, 2020.

\bibitem{amiri2020machine}
M.~M. Amiri and D.~G{\"u}nd{\"u}z, ``Machine learning at the wireless edge:
  Distributed stochastic gradient descent over-the-air,'' \emph{IEEE
  Transactions on Signal Processing}, vol.~68, pp. 2155--2169, 2020.

\bibitem{tran2019federated}
N.~H. Tran, W.~Bao, A.~Zomaya, N.~M. NH, and C.~S. Hong, ``Federated learning
  over wireless networks: Optimization model design and analysis,'' in
  \emph{2019 IEEE Conference on Computer Communications (INFOCOM)}.\hskip 1em
  plus 0.5em minus 0.4em\relax IEEE, 2019, pp. 1387--1395.

\bibitem{kang2019incentive}
J.~Kang, Z.~Xiong, D.~Niyato, S.~Xie, and J.~Zhang, ``Incentive mechanism for
  reliable federated learning: A joint optimization approach to combining
  reputation and contract theory,'' \emph{IEEE Internet of Things Journal},
  vol.~6, no.~6, pp. 10\,700--10\,714, 2019.

\bibitem{nishio2019client}
T.~Nishio and R.~Yonetani, ``Client selection for federated learning with
  heterogeneous resources in mobile edge,'' in \emph{2019 IEEE International
  Conference on Communications (ICC)}.\hskip 1em plus 0.5em minus 0.4em\relax
  IEEE, 2019, pp. 1--7.

\bibitem{zhan2020incentive}
Y.~Zhan and J.~Zhang, ``An incentive mechanism design for efficient edge
  learning by deep reinforcement learning approach,'' in \emph{IEEE INFOCOM
  2020-IEEE Conference on Computer Communications}.\hskip 1em plus 0.5em minus
  0.4em\relax IEEE, 2020, pp. 2489--2498.

\bibitem{zhan2020learning}
Y.~Zhan, P.~Li, Z.~Qu, D.~Zeng, and S.~Guo, ``A learning-based incentive
  mechanism for federated learning,'' \emph{IEEE Internet of Things Journal},
  2020.

\bibitem{globecom}
Q.~Hu, F.~Li, X.~Zou, and Y.~Xiao, ``Correlated participation decision making
  for federated edge learning,'' in \emph{IEEE GLOBECOM 2020-IEEE Global
  Communications Conference}.\hskip 1em plus 0.5em minus 0.4em\relax IEEE,
  2020.

\bibitem{chen2018my}
I.~Chen, F.~D. Johansson, and D.~Sontag, ``Why is my classifier
  discriminatory?'' \emph{Advances in Neural Information Processing Systems},
  pp. 3539--3550, 2018.

\bibitem{johnson2018predicting}
M.~Johnson, P.~Anderson, M.~Dras, and M.~Steedman, ``Predicting accuracy on
  large datasets from smaller pilot data,'' \emph{Proceedings of the 56th
  Annual Meeting of the Association for Computational Linguistics}, pp.
  450--455, 2018.

\bibitem{wang2019measure}
G.~Wang, C.~X. Dang, and Z.~Zhou, ``Measure contribution of participants in
  federated learning,'' in \emph{2019 IEEE International Conference on Big Data
  (Big Data)}.\hskip 1em plus 0.5em minus 0.4em\relax IEEE, 2019, pp.
  2597--2604.

\bibitem{song2019profit}
T.~Song, Y.~Tong, and S.~Wei, ``Profit allocation for federated learning,'' in
  \emph{2019 IEEE International Conference on Big Data (Big Data)}.\hskip 1em
  plus 0.5em minus 0.4em\relax IEEE, 2019, pp. 2577--2586.

\bibitem{hurwicz2006designing}
L.~Hurwicz and S.~Reiter, \emph{Designing economic mechanisms}.\hskip 1em plus
  0.5em minus 0.4em\relax Cambridge University Press, 2006.

\bibitem{algorithmic2007}
N.~Nisan, T.~Roughgarden, E.~Tardos, and V.~V. Vazirani, \emph{Algorithmic game
  theory}.\hskip 1em plus 0.5em minus 0.4em\relax Cambridge University Press,
  2007.

\bibitem{lecun1998gradient}
Y.~LeCun, L.~Bottou, Y.~Bengio, and P.~Haffner, ``Gradient-based learning
  applied to document recognition,'' \emph{Proceedings of the IEEE}, vol.~86,
  no.~11, pp. 2278--2324, 1998.

\end{thebibliography}

\begin{IEEEbiography}[{\includegraphics[width=1in,height=1.25in,clip,keepaspectratio]{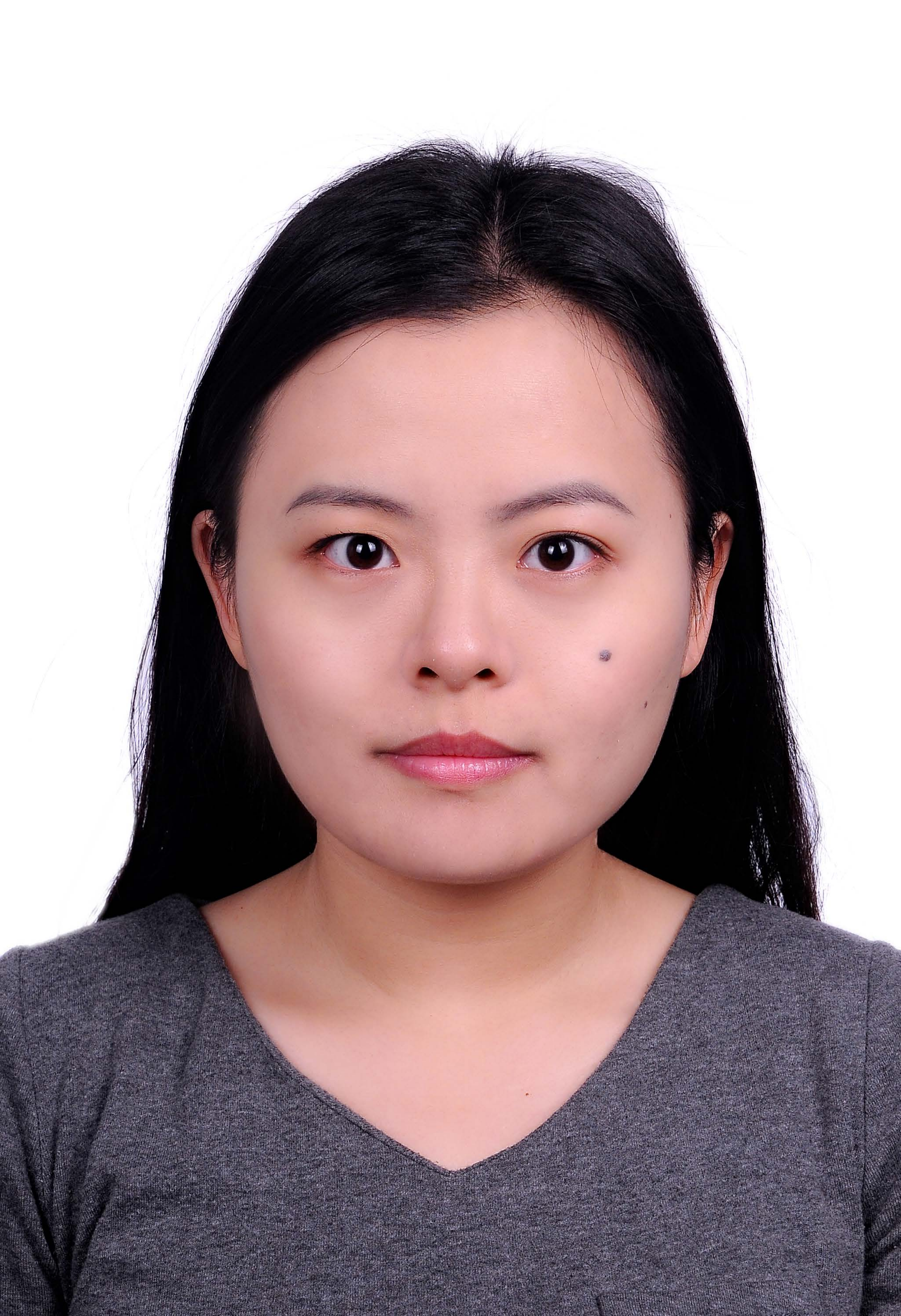}}]{Qin Hu} received her Ph.D. degree in Computer Science from the George Washington University in 2019. She is currently an Assistant Professor with the Department of Computer and Information Science, Indiana University-Purdue University Indianapolis (IUPUI). Her research interests include wireless and mobile security, mobile edge computing, and blockchain.
\end{IEEEbiography}

\begin{IEEEbiography}[{\includegraphics[width=1in,height=1.25in,clip,keepaspectratio]{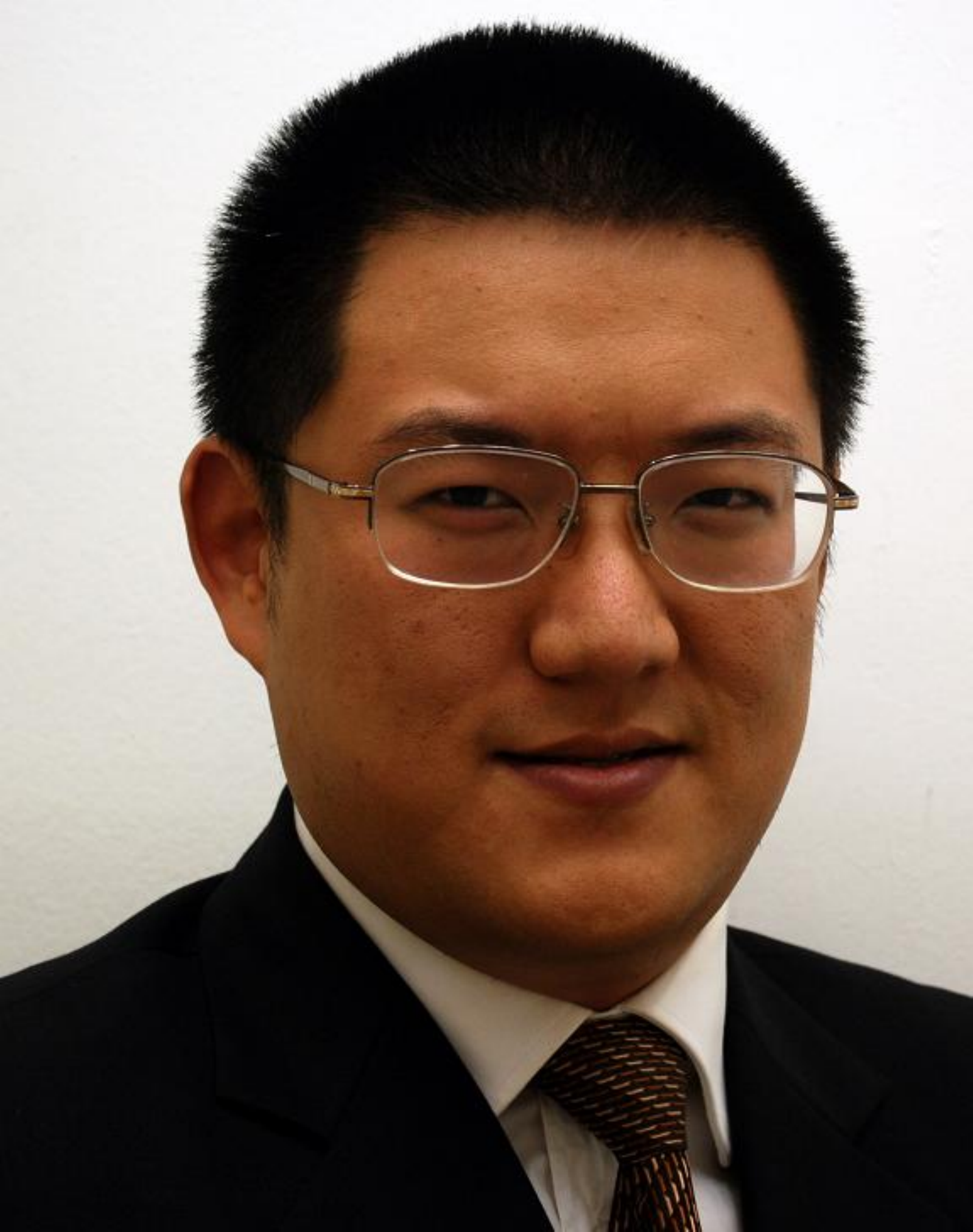}}]{Feng Li} received the Ph.D. degree in Computer Science from Florida Atlantic University in August 2009. He is an Associate Professor of Computer and Information Technology with Indiana University–Purdue University Indianapolis (IUPUI). He has published more than 50 papers in top conferences, including the INFOCOM and ICDCS. His research interests include the areas of cybersecurity and trust issues, cloud, and mobile computing.
\end{IEEEbiography}

\begin{IEEEbiography}[{\includegraphics[width=1in,height=1.25in,clip,keepaspectratio]{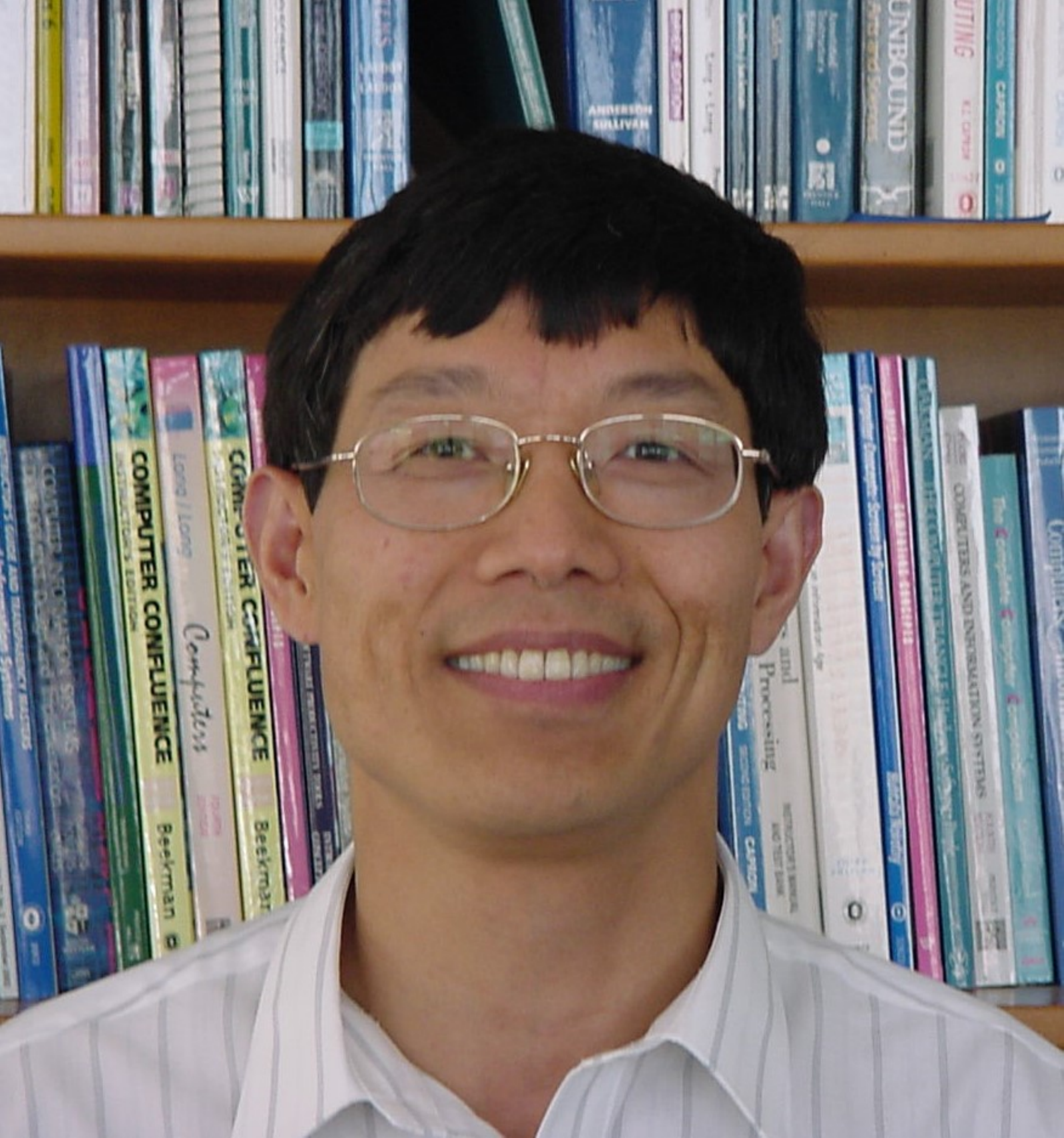}}]{Xukai Zou} received the Ph.D. degree in Computer
Science from the University of Nebraska–Lincoln,
Lincoln, NE, USA.
He is currently a Professor with the
Department of Computer and Information Sciences,
Indiana University–Purdue University Indianapolis,
Indianapolis, IN, USA. His research has been supported by NSF, the Department of Veterans Affairs,
and Industry such as Cisco, San Jose, CA, USA.
His current research interests include cryptography,
communication networks and security, secret sharing, health and personal genomic security and privacy, design and analysis of
algorithms, and image and data compression.
\end{IEEEbiography}

\begin{IEEEbiography}[{\includegraphics[width=1in,height=1.25in,clip,keepaspectratio]{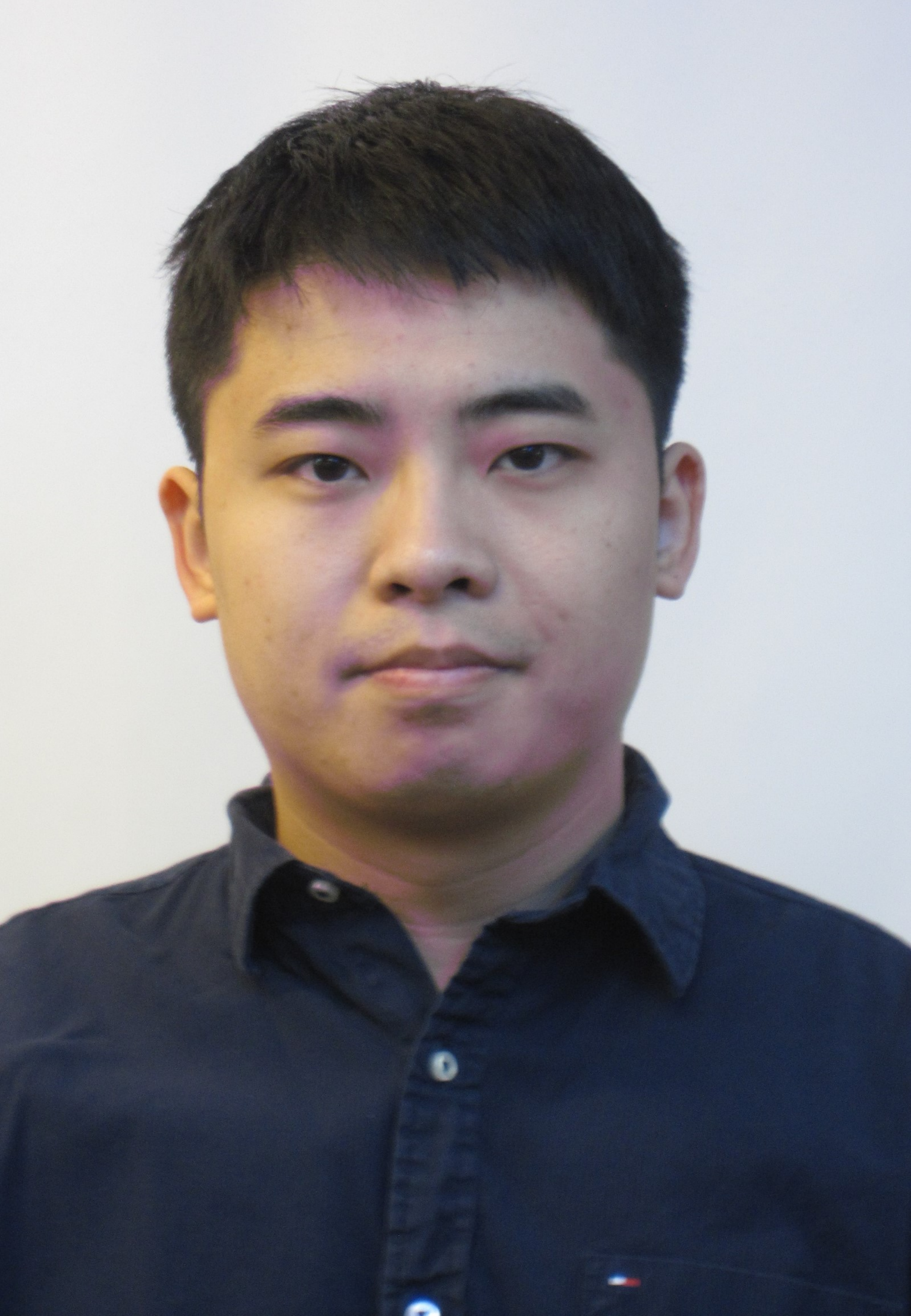}}]{Yinhao Xiao} received his Ph.D. degree in Computer Science from the George Washington University in 2019. He is currently a Faculty Member with the School of Information Science, Guangdong University of Finance and Economics, Guangzhou, China. 
His current research interests include the IoT security, smartphone security, and binary security.
\end{IEEEbiography}

\end{document}